\documentclass[10pt,a4paper]{article}
\usepackage[T2A]{fontenc}
\usepackage[utf8]{inputenc}
\usepackage[english]{babel}
\usepackage{amsthm}
\usepackage{amsbib}
\usepackage{amsfonts,amssymb,mathrsfs,amscd}

\long\def\tmpremove#1\endtmpremove{}
\long\def\tmpadd#1\endtmpadd{}

\usepackage{graphicx}
\usepackage{hyperref}
\usepackage{pgf,tikz}
\usepackage{mathtools, nccmath}
\usetikzlibrary{arrows}
\usepackage{wrapfig}
\usepackage{diagbox}

\numberwithin{equation}{section}
\theoremstyle{definition}
\newtheorem{definition}{Definition}
\theoremstyle{theorem}
\newtheorem{theorem}{Theorem}
\newtheorem{lemma}{Lemma}[section]
\newtheorem{proposition}{Proposition}

\theoremstyle{definition}
\newtheorem{remark}{Remark}
\newcommand{\bluer}[1]{#1}
\newcommand{\bluevar}[1]{#1}

\newcommand{\remove}{}
\begin{document}

\title{Feynman checkers: external electromagnetic field and asymptotic properties}
%\address{National Research University Higher School of Economics (Faculty of Math)}
%\email{FedorO57@yandex.ru}
%второй автор
%\author[I.\,I.~Ivanov]{И.\,И.~Иванов}
%\address{}
%\email{}

\date{}
\author{Fedor Ozhegov\\ 
Faculty of Mathematics, HSE University, \\
FedorO57@yandex.ru}
%\udk{}
\maketitle
\begin{abstract}
We study Feynman checkers, one of the most elementary models of electron motion. It is also known as a one-dimensional quantum walk or an Ising model at an imaginary temperature. We add the simplest nontrivial electromagnetic field and find the limits of the resulting model for small lattice step and large time, analogous to the results by J.~Narlikar from 1972 and G. Grimmet - S. Jason - P. Scudo from the 2000s. It turns out that the limits in the model with the added field are obtained from the ones without field by mass renormalization. Also we find an exact solution of the resulting model. \remove 
\end{abstract}

%\begin{keywords}
\textbf{Key words:}
Feynman checkers, Dirac equation, quantum walk, lattice gauge theory, renormalization
%\end{keywords}

\textbf{Mathematics Subject Classifcation:}  82B20, 33C45, 81T25
\markright{Feynman Checkers: external electromagnetic field}

%%
%% ===========================================================================
%%

\section{Introduction}
Feynman checkers is one of the simplest models of electron motion. It was invented by R.~Feynman and published in 1965 in~\cite{feyn} (see~\cite{Konno-20, article} for recent surveys). A. Ambainis et.~al established brilliant results about it in~\cite{Amb}. Although they studied the one-dimensional quantum walk and the Hadamar walk, their model was equivalent to Feynman checkers. Nowadays, quantum walks are actively developed --- see \cite{Cedzich-et-al-2018, Cedzich-Werner-2021, MasMay-AkiSuz-2019, Skopenkov-Ustinov-22} for most recent results. We study Feynman checkers with an external elecromagnetic field and establish two phenomena: mass renormalization (Theorems \ref{wavelimit}, \ref{large-time-lim}) and spin precession (Theorem \ref{chir-rev-f}). See Fig.~\ref{QC} for a quantum-computer implementation. \bluer{Cf. \cite[Figure 6]{article} for the model without field.}

\begin{figure}
\includegraphics[width=\linewidth]{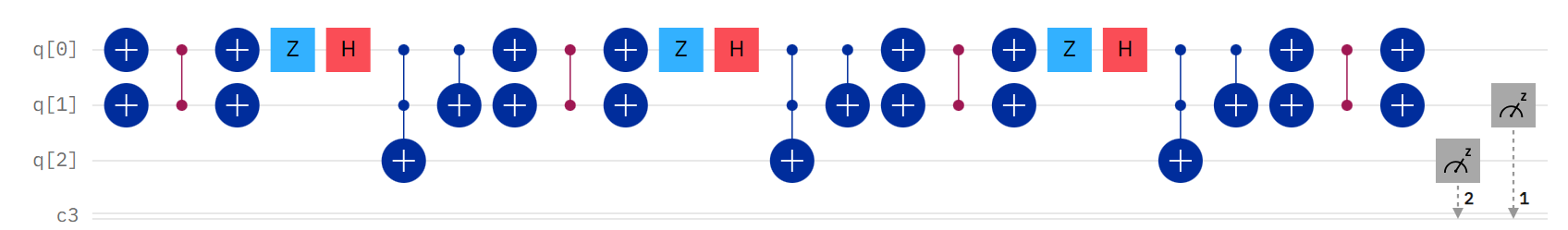}\\
\begin{tabular}{lccc}
\includegraphics[width=.5\linewidth]{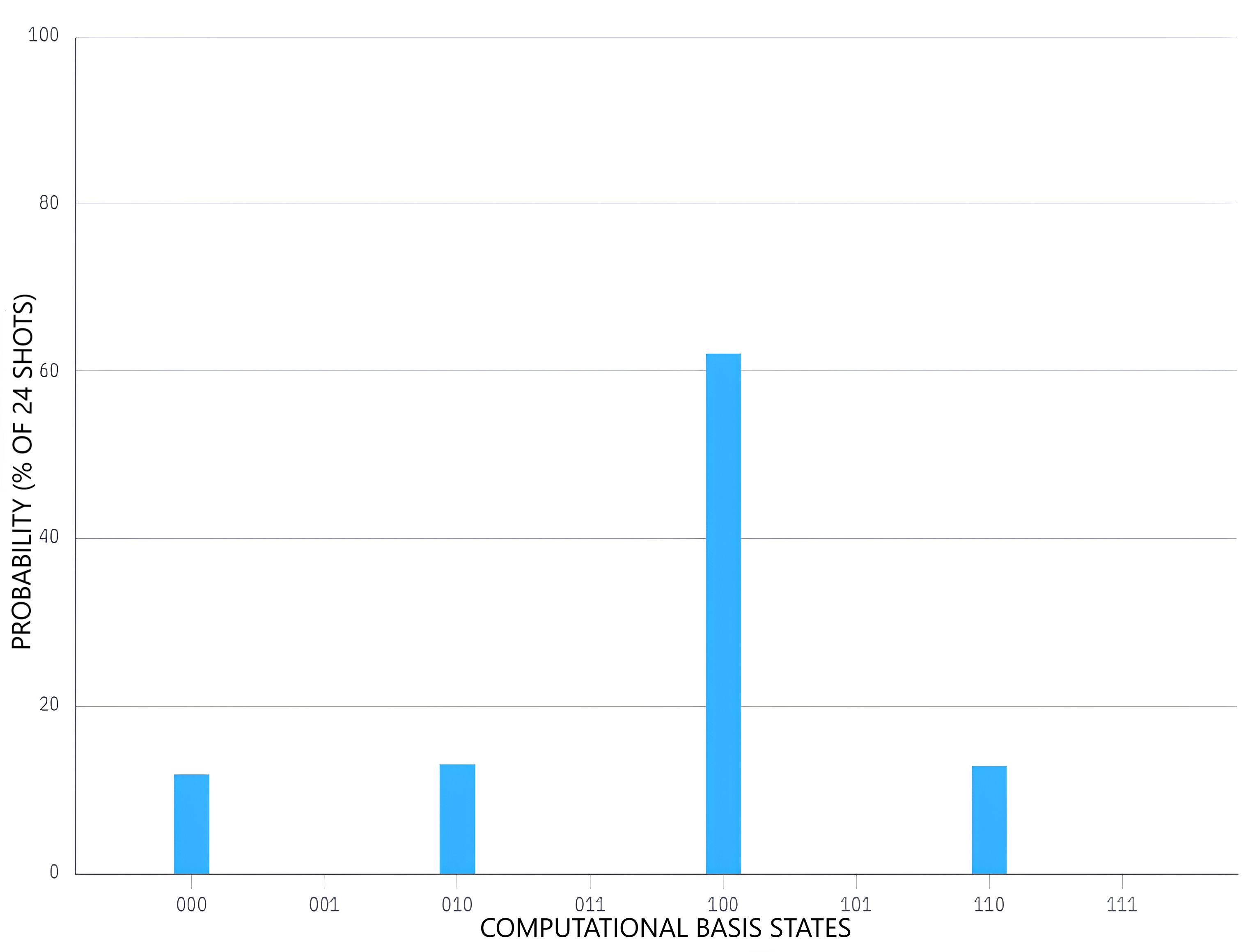} 
\includegraphics[width=.5\linewidth]{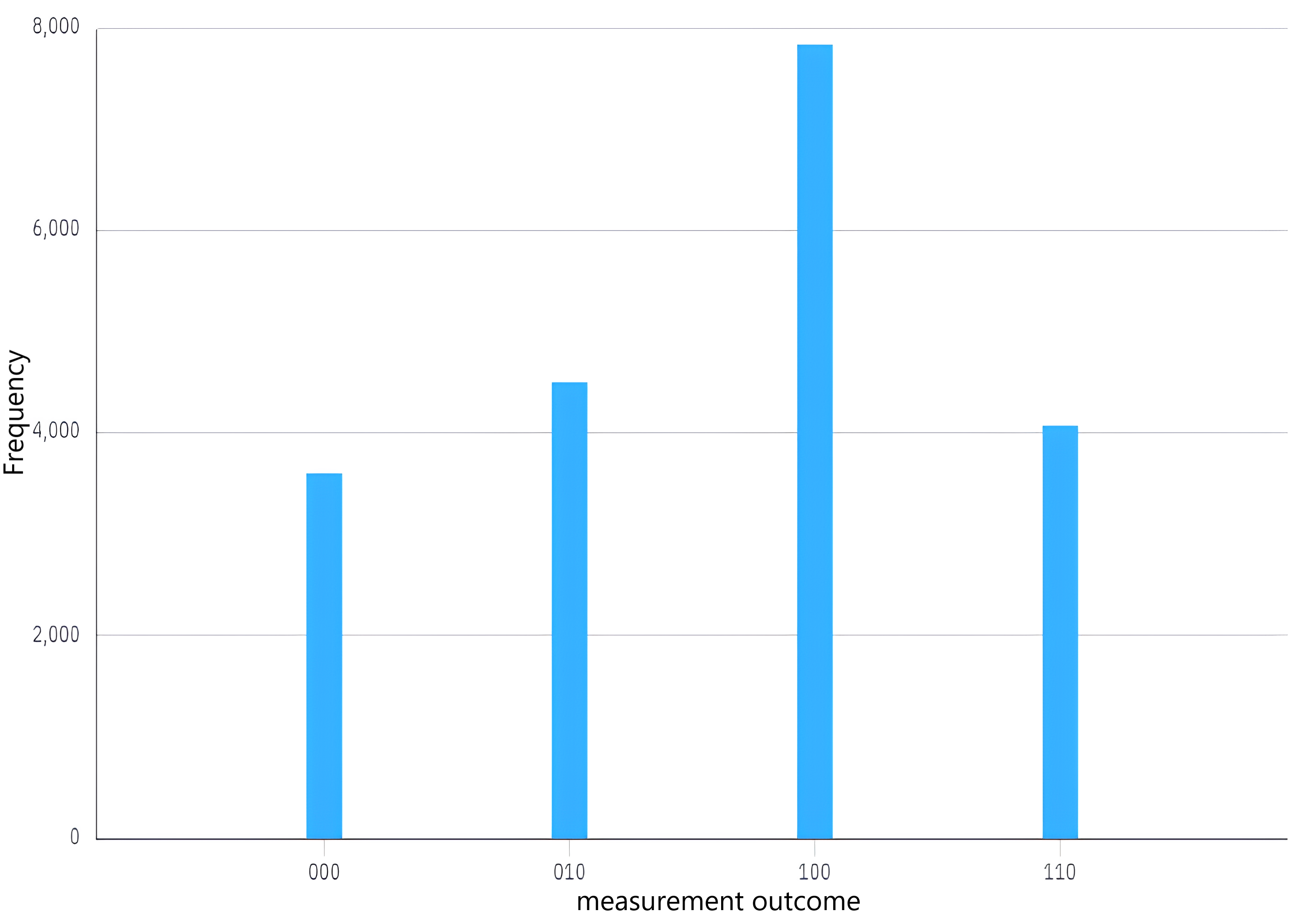}
\end{tabular}
\caption{Implementation of Feynman checkers with a homogeneous electromagnetic field on a quantum computer using quantum circuit language (top). The output is a random bit-string coding electron position $x$ at time $t= 4$. The strings 000, 010, 100, 110 code $x = 4, 2, 0, -2$ respectively. Distribution of $x$ (bottom-left) and a histogram for quantum computer IBM-Lima (bottom-right). Here the red operator is CZ with the control q[1] and the target q[0].} 
\label{QC}
\end{figure}

Let us first give a survey of known asymptotic results in Feynman checkers and then state the new ones. All the notions used in the following theorems are going to be defined precisely in Section~\ref{opredel}.

In 1972 J. Narlikar found the continuum limit of the model. This limit coincides with the well-known expression, obtained by solving Dirac's equation on the line. 
\begin{theorem}[{\cite[Theorem 8]{article}}]
\label{problim}
Assume $m, \varepsilon>0 ,|x| < t$, where $x/2\varepsilon, t/2\varepsilon\in \mathbb{Z}$. Then, on the $2$-dimensional square lattice of step $\varepsilon$, the divided by $4\varepsilon^2$  probability to find an electron of mass $m$ at the point $(x,t)$, if it was emitted from $(0,0)$, as $\varepsilon \to 0$ tends to
 $$\frac{m^2}{4}\left( J_0\left(m \sqrt{t^2-x^2}\right)^2 + \frac{t+x}{t-x}J_1\left(m \sqrt{t^2-x^2}\right)^2 \right).$$

\end{theorem}

Here $J_0(z) := \sum\limits_{j=0}^{\infty}(-1)^j\frac{(z/2)^{2j}}{(j!)^2}$ and $J_1(z) := \sum\limits_{j=0}^{\infty}(-1)^j\frac{(z/2)^{2j+1}}{(j!)(j+1)!}$ are \emph{Bessel functions} of the first kind of orders 0 and 1 respectively. The terms containing $J_1$ and $J_0$ correspond to the probabilities to find the electron with the original and reversed \emph{chirality} respectively (see \cite[\S4]{article}).

Elementary mathematical proof of this theorem is given in~\cite[Appendix A]{article}. For the case of smooth enough initial conditions an analogous result was proved in~\cite{MasMay-AkiSuz-2019}.

In the early 2000s N. Konno and G. Grimmet--S. Janson--P. Scudo found the large-time limit of the model, namely, the limiting distribution of the electron position.

\begin{theorem}[{\cite[Theorem 1]{Grimmet-Janson-Scudo-04}}]
\label{wl-Grimmet}
Assume that $m,\varepsilon,t>0, v\in \mathbb{R}$, where $t/\varepsilon\in \mathbb{Z}$. Then on the $2$-dimensional square lattice of step $\varepsilon$, the probability to find the electron of mass $m$, emitted from the point $0$ at moment $0$, to the left from the point $vt$ at the moment $t$, as $t\to \infty$ tends to 
$$
F(v):=
    \begin{cases}
        0 ,& \text{if } v < -\frac{1}{\sqrt{1+m^2\varepsilon^2}};\\
        \frac{1}{\pi}\arccos{\frac{1-(1+m^2\varepsilon^2)v}{\sqrt{1+m^2\varepsilon^2}(1-v)}}, &\text{if } |v| \le \frac{1}{\sqrt{1+m^2\varepsilon^2}};\\
        1 , &\text{if } v > \frac{1}{\sqrt{1+m^2\varepsilon^2}}.
    \end{cases}
$$
\end{theorem}

A short proof of this theorem can be found in \cite[\S 12.1]{article}.

This result has numerous variations and generalizations \cite[\S 3.2]{article}, \cite{Kuyanov-Slizkov}. Let us state one variation by I. Bogdanov (cf. \cite[Theorem 2]{IB-20}).

\begin{theorem}[{\cite[Theorem 2]{IB-20}}]
\label{chiral-rev} Assume $0 \le m\varepsilon \le 1$. Then on the $2$-dimensional square lattice of step $\varepsilon$, the probability to find an electron, emitted from $0$ at the moment $0$, with the reversed chirality at the moment $t$, tends to $\frac{m\varepsilon}{2\sqrt{1+m^2\varepsilon^2}}$ as $t \to \infty$ so that $t/\varepsilon$ is an integer.
\end{theorem}

Recently P. Zakorko \cite{Zakorko} found a uniform approximation of the wave function by Airy function extending earlier results by  T. Sunada - T. Tate from \cite{Sunada2012}. She used the method of A. Anikin et. al from \cite{Anikin2019}.

The following \bluer{(unpublished)} theorem is formulated in terms of \emph{Airy function}
\begin{align*}
    \mathrm{Ai}(\lambda)  := \frac{1}{\pi} \int\limits_0^{\infty} \cos \left(\frac{x^3}{3} + \lambda x\right) dx.
\end{align*}
\begin{theorem}[{Cf. \cite[Theorem 1]{Zakorko}}]\label{prob_Zak}
Assume $|x| < t/\sqrt{2}$, where $x,t\in \mathbb{Z}$ and $x+t$ is odd. Then, on the $2$-dimensional square lattice of unit step, the probability to find an electron of unit mass at the point $(x,t+1)$, if it was emitted from $(0,0)$, equals:
\begin{align*}
    &\left( \frac{4\theta(x/t)}{1 -2 (x/t)^2} \right)^{1/2} \left(\frac{1}{t}\right)^{2/3}\mathrm{Ai}\left(-\theta(x/t)t^{2/3}\right)^2 + \\
    &+\frac{t+x-1}{t-x\bluevar{+}1}\left(\frac{4\theta((x-1)/t)}{1-2((x-1)/t)^2}\right)^{1/2}\left(\frac{1}{t}\right)^{2/3} \mathrm{Ai}\left(-\theta((x-1)/t)t^{2/3}\right)^2 + O\left(\frac{1}{t}\right),
\end{align*}
where
\begin{align*}
    \theta(v) := \left(\frac{3}{2}\left(-|v| \arccos\frac{|v|}{\sqrt{1-v^2}} + \arccos \frac{1}{\sqrt{2-2v^2}}\right)\right)^{2/3}.
\end{align*}
\end{theorem}

Hereafter the notation $f(t) = g(t) + O(h(t))$ means that there exists a constant $C$ (independent of $t$) such that for each $t$ satisfying the conditions of the theorem, we have  $|f(t) - g(t)| \le Ch(t)$. \remove

\textbf{Feynman checkers with a field.}

There are many modifications of this model. For example, in the 1990s an electromagnetic field was added to the model in ~\cite{shul} and  \cite{Ord}. The resulting model is actively studied nowadays (see \cite{Cedzich-Werner-2021, Cedzich-et-al-2013}). This modification is equivalent to inhomogeneuos quantum walk; see the recent survey~\cite{Konno-20} and also \cite{Skopenkov-Ustinov-22, Bourgain2000, Dmitriev-22}. Despite the extensive literature, no asymptotic formulae for the wave function have been known before for any electromagnetic field (gauge nonequivalent to the zero field).

In the present work, Theorems~\ref{problim}--\ref{prob_Zak} are translated to the case of the simplest non-trivial electromagnetic field (see Figures ~\ref{Limit},~\ref{fig: chir-rev-f} and Theorems \ref{fproblim}, \ref{wl-intr}, \ref{chir-rev-f}, \ref{Airy} respectively).

\begin{figure}[h]
\begin{tabular}{cc} 
\includegraphics[width=0.47\linewidth]{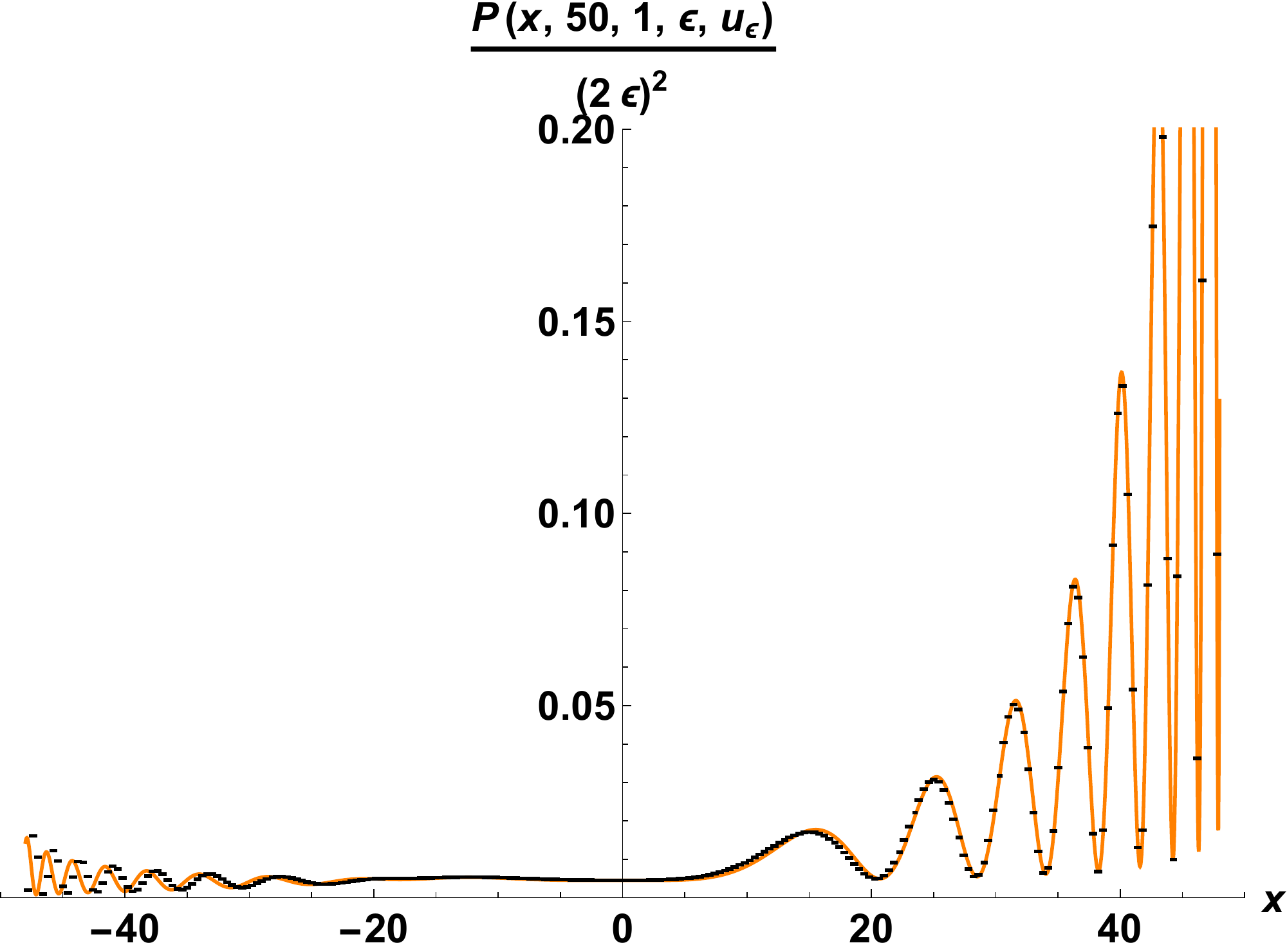}
\includegraphics[width=.47\linewidth]{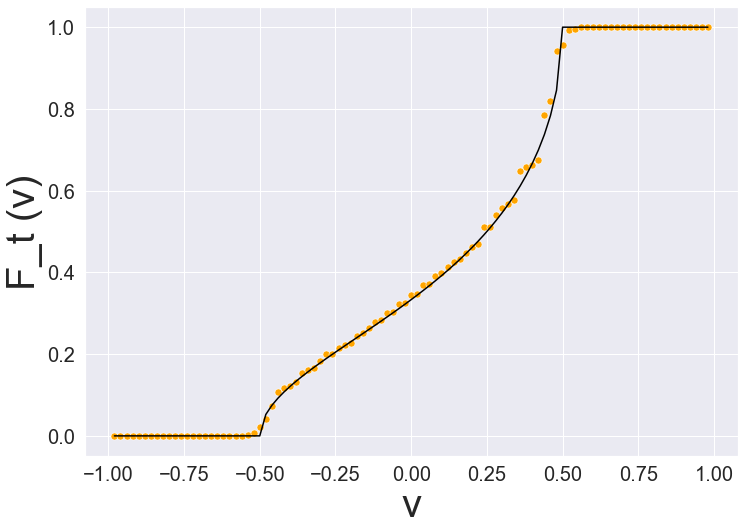}
\end{tabular}
\caption{{Left: The divided by $4\varepsilon^2$ probability to find an electron of mass $1$ at the point $(x,t)$, if it was emitted from the point $(0,0)$ and moved in the homogeneous field, for $\varepsilon=0.1$ (dashes) and $\varepsilon \to 0$ (curve). Here $t=50$ and $5x \in \mathbb{Z}$.  See Definition~\ref{def1} and Theorem~\ref{fproblim}. Right: The plot of the probability $F_t(v)$ to find an electron of mass 1, if it was emitted from the origin and moved in the homogeneous electromagnetic field, to the left from the point $vt$ at the moment $t$, for $|v|\le 1$ and $t=100$ (dots). The limit of this probability as $t\to\infty$ (curve). See Definition \ref{def1} and Theorem~\ref{wl-intr}.}}
\label{Limit}
\end{figure}

\begin{theorem}
\label{fproblim}
Assume  $m, \varepsilon>0, |x| < t,$ where $x/4\varepsilon, t/4\varepsilon \in \mathbb{Z}$. Then on the $2$-dimensional square lattice of step $\varepsilon$, the divided by $4\varepsilon^2$ probability to find an electron of mass $m$ at the point $(x,t)$, if it was emitted from the origin and moved in the homogeneous electromagnetic field, as $\varepsilon \to 0$ tends to $$\frac{m^2}{4}\left( J_0\left(m \sqrt{\frac{t^2-x^2}{2}}\right)^2 + 2\frac{t+x}{t-x}J_1\left(m \sqrt{\frac{t^2-x^2}{2}}\right)^2 \right).$$
\end{theorem}

In the course of the proof of this theorem, for the first time we obtain an ''explicit'' expression for the wave function of an electron moving in the homogeneous electromagnetic field, in the discrete model, that is, its exact solution (see Proposition~\ref{formula}). In contrast to the initial model, this expression cannot be derived from simple combinatorial ideas. Having got the ''explicit'' expression, we apply the method of \bluer{the} proof of Theorem~\ref{problim} above (see \cite[Appendix B]{article}).

The motion in such electromagnetic field is the simplest example of an inhomogeneous quantum walk, with the inhomogeneity having a period of 2 in space and time. Periodic quantum walks were studied in~\cite[\S 3.1]{Cantero-etal-12} but no asymptotic formulae have been known for them before.
%В литературе эта модель так же известна как блуждания Адамара(~\ref{Konno}, ~\ref{Ambainic} и другие работы).

\begin{theorem}
\label{wl-intr}
Assume $m,\varepsilon,t>0, v\in \mathbb{R}$ and $t/\varepsilon \in \mathbb{Z}$. Then on the $2$-dimensional square lattice of step $\varepsilon$, the probability to find an electron of mass $m$, if it was emitted from $0$ at the moment $0$ and moved in a homogeneous electromagnetic field, to the left from the point $vt$ at the moment $t$, as $t\to \infty$ tends to 
$$
F(v):=
    \begin{cases}
        0 ,& \text{if } v < -\frac{1}{1+m^2\varepsilon^2};\\
        \frac{1}{\pi}\arccos{\frac{1-(1+m^2\varepsilon^2)^2v}{(1+m^2\varepsilon^2)(1-v)}}, &\text{if } |v| \le \frac{1}{1+m^2\varepsilon^2};\\
        1 , &\text{if } v > \frac{1}{1+m^2\varepsilon^2}.
    \end{cases}
$$
\end{theorem}

Thus, the formula in the model with the added field (Theorem~\ref{wl-intr}) can be obtained from the one without field (Theorem~\ref{wl-Grimmet}) by mass \emph{renormalization}:
$$
(1+m^2\varepsilon^2)^2 = 1 +m_0^2\varepsilon^2,
$$
where $m$ is the mass in the model with the field, and $m_0$ is the one in the model without field. Tending $\varepsilon$ to $0$, we obtain the following relation:
$$
1+2m^2\varepsilon^2 + o(\varepsilon^2)=1 + (m_0\varepsilon)^2.
$$
Thus, $m\sim m_0/\sqrt{2}$ as $\varepsilon \to 0$. It is precisely the relation between the arguments of the Bessel functions in Theorems \ref{problim} and \ref{fproblim}.

\begin{theorem}[{See Fig. \ref{fig: chir-rev-f}}]
\label{chir-rev-f}
Assume $m,\varepsilon >0$. Then on the $2$-dimensional square lattice of step $\varepsilon$, the probability to find an electron, emitted from the point $0$ at the moment $0$ and moved in the homogeneous field, with the reversed chirality at the moment $t$, as $t \to \infty$ so that the parity of~$t/\varepsilon$ is fixed tends to
$$
\begin{cases}
\frac{m\varepsilon}{(1+m^2\varepsilon^2)\sqrt{2 + m^2\varepsilon^2}}, & \text{ if } \frac{t}{\varepsilon} \equiv_2 1;\\
\frac{m\varepsilon}{\sqrt{2 + m^2\varepsilon^2}}, & \text{ if } \frac{t}{\varepsilon} \equiv_2 0,
\end{cases}
$$
\end{theorem}

\begin{figure}[h]
\begin{tabular}{cc}
    \includegraphics[width=0.47\linewidth]{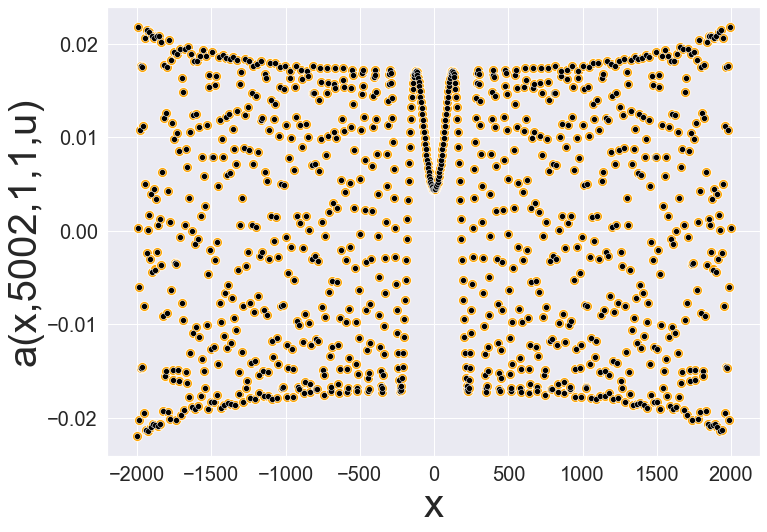}
    \includegraphics[width=0.47\linewidth]{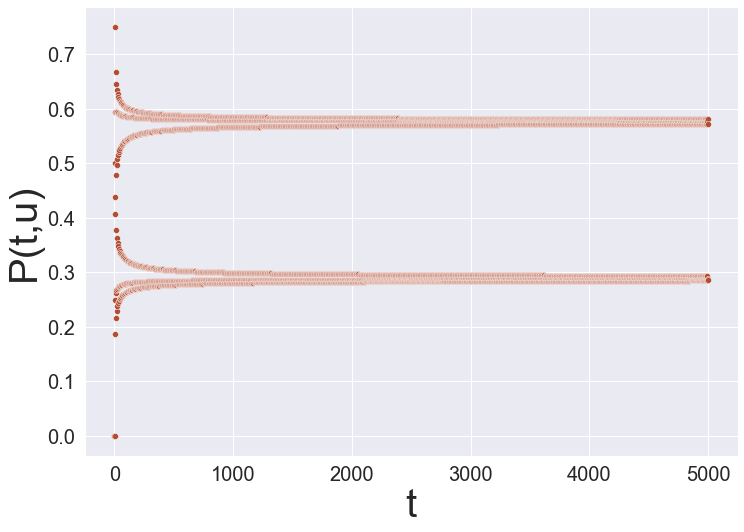}
\end{tabular}
    \caption{\footnotesize{Left: The plot of the wave function component $a_1(x, 5002,1,1,u_1)$ for $x$ divisible by 4 is shown in black and the approximation given by Theorem \ref{Airy} is shown in orange. Right: The graph of the probability $P(t,u)$ of chirality reversal for an electron of mass 1, emitted from the origin and moving in a homogeneous electromagnetic field. See Definition \ref{def1} and Theorem~\ref{chir-rev-f}.}}
\label{fig: chir-rev-f}
\end{figure}

This result demonstrates \emph{spin precession in an electromagnetic field}: the probability to find the electron with the reversed chirality tends to a periodic function (rather than a constant as in the case of the zero field; cf. Theorem \ref{chiral-rev} and \cite[Example 5]{article}). This solves Problem 11 from~\cite{article} and proves Hypothesis 1 from~\cite{IB-20} (to be more precise, in those papers the field was not equal to \bluer{our} $u_{1}$ but only gauge equivalent to it, but this leads to the same probabilities).

In Theorem \ref{Airy} (see Figure \ref{fig: chir-rev-f}) we \bluer{announce} a uniform asymptotic formula in terms of \bluer{the} Airy function.  This is another example of the mass renormalization.

In \S 3 we state Theorems \ref{fproblim}, \ref{wl-intr}, \ref{chir-rev-f} in a stronger form, and in the subsequent sections we prove the results from \S 3. Pure calculations are put in \ref{first-proof} and \ref{IntFormProof}.

\section{Definitions}
\label{opredel}
First we give \bluer{an} informal description of Feynman checkers, then of the modification with added electromagnetic field, and finally we give a precise definition. The major part of this introductory section is borrowed from~\cite{article}.
\par
Fix $m \ge 0$ called the \emph{mass} of the electron. Consider the infinite checkerboard made of squares $\varepsilon\times\varepsilon$. The checker moves to the diagonal-neighboring squares, either upwards-left or upwards-right. To each path $s$ of the checker, we assign a vector $a(s)$ as follows (Fig.~\ref{Checker-paths}). Initially this vector is directed upwards and has unit length. While the checker moves straightly the vector remains unchanged, and after each turn of the checker it is rotated by $90^{\circ}$ clockwise and multiplied by $m\varepsilon$. At the end of the motion the vector is shrinked by a factor of $(1+m^2\varepsilon^2)^{\frac{t/\varepsilon - 1}{2}}$, where $t/\varepsilon$ is the total number of moves. The resulting vector is $a(s)$.

Denote by $a(x,t,m,\varepsilon) :=\sum_{s}a(s)$ the sum over all the checker paths from the square $(0,0)$ to the square $(x,t)$, starting from the upwards-right move. The length square of the vector $a(x,t,m,\varepsilon)$ is called \emph{the probability to find an electron in the square} $(x,t)$, \emph{if it was emitted from the origin}, and the vector itself is called \emph{the arrow} or \emph{the wave function}. For example, in a Figure~\ref{Checker-paths}\bluer{(from \cite{article})} to the right we have $a(s_0)=(1/2,0), a(s)=(0,-1/2)$, аnd $a(1,3,1,1) = (1/2,-1/2)$. 
\begin{figure}[htbp]
%\vspace{-0.5cm}
\begin{center}
\begin{tabular}{l}
%\hspace{0.5cm}
\end{tabular}
\begin{tabular}{lccc}

\includegraphics[width=1.4cm]{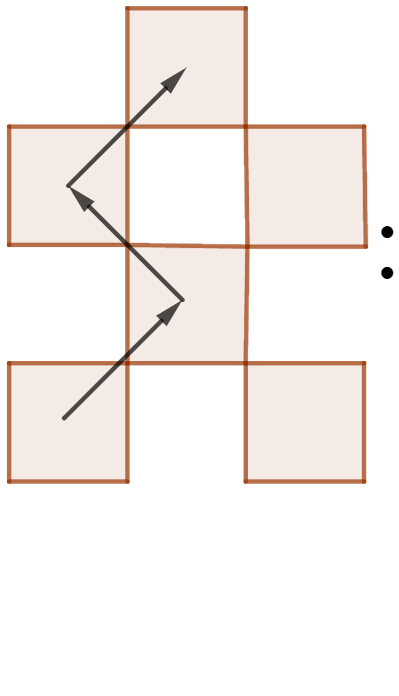} &
\includegraphics[width=1.5cm]{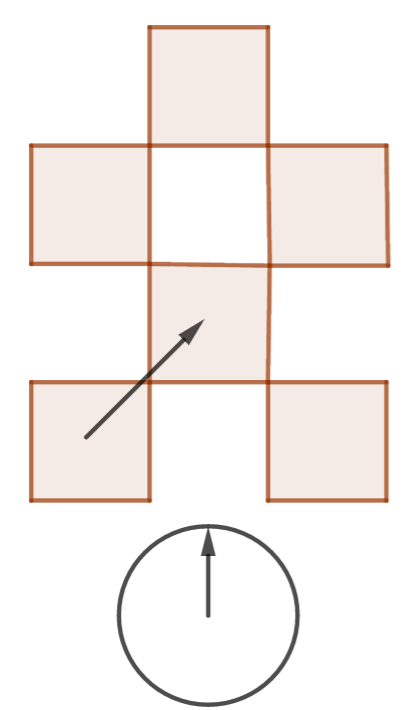} &
\includegraphics[width=1.5cm]{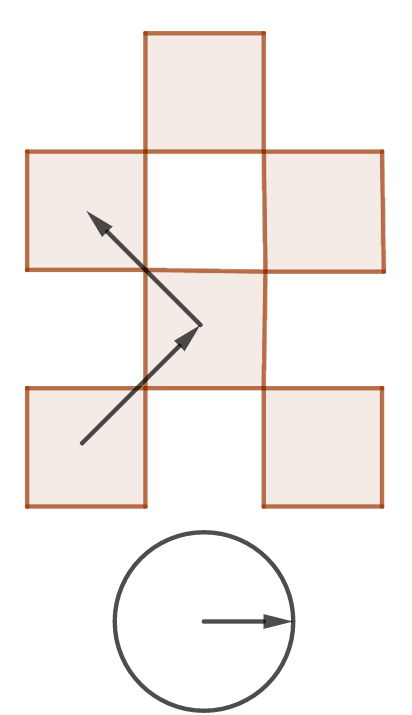} &
\includegraphics[width=1.5cm]{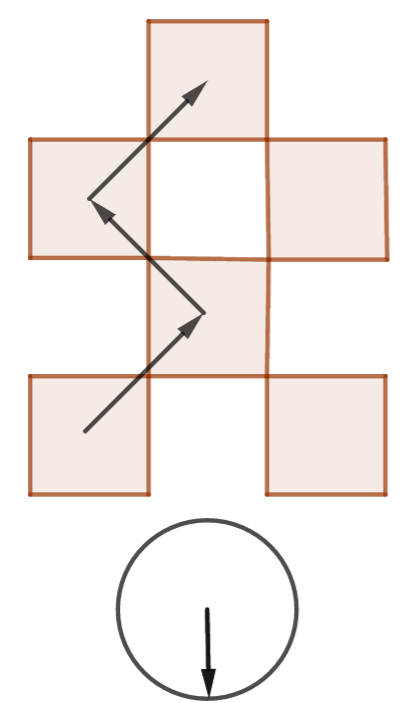} \\
\includegraphics[width=1.5cm]{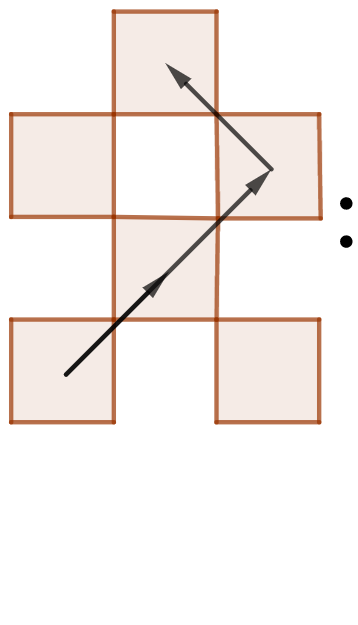} &
\includegraphics[width=1.5cm]{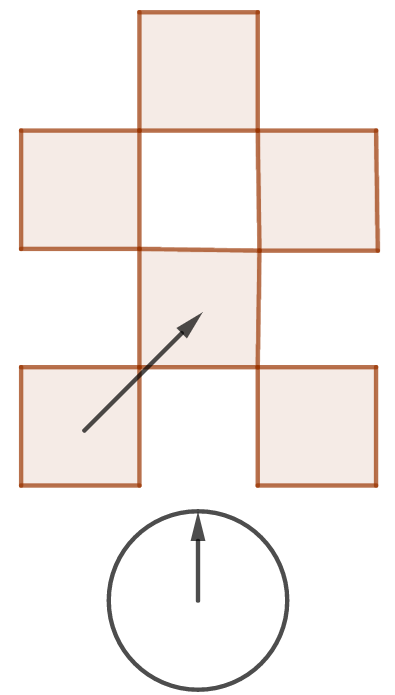} &
\includegraphics[width=1.5cm]{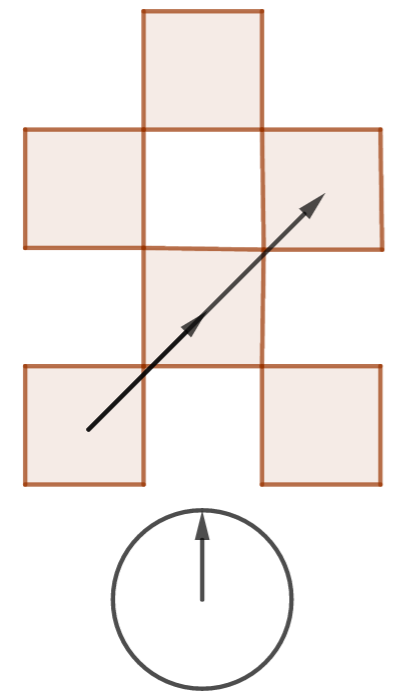} &
\includegraphics[width=1.5cm]{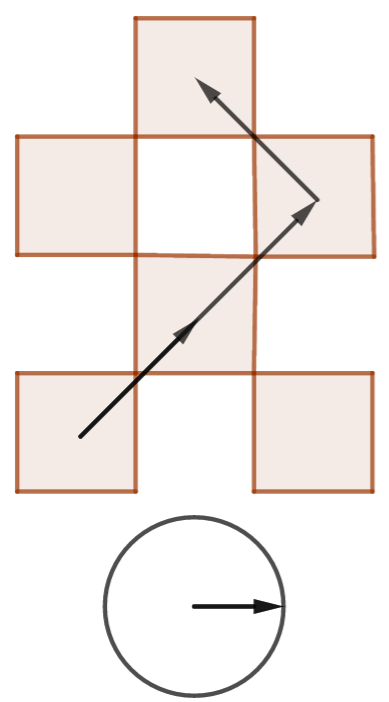}
\end{tabular}
\begin{tabular}{r}
%\hspace{0.5cm}
\includegraphics[width=3.5cm]{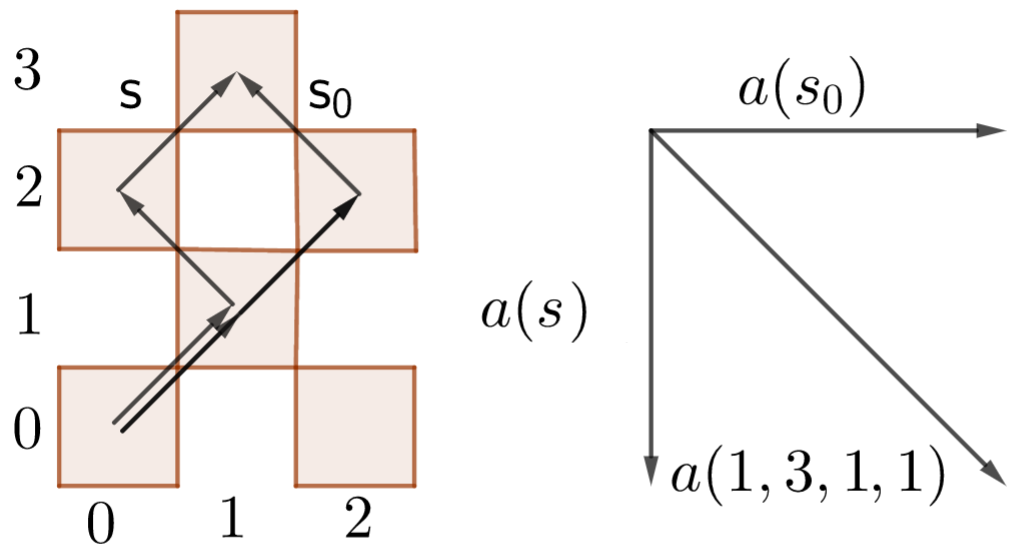}
\end{tabular}
\end{center}
\vspace{-0.5cm}
\caption{{Checker paths (left). Vectors, assigned to paths (right).}}
\label{Checker-paths}
\vspace{-0.3cm}
\end{figure}
\par
%\textbf{Определение 2. }

%Зафиксируем $\varepsilon > 0$ и $m\ge 0$ --- \textit{шаг решетки} и \textit{массу частицы} соответственно. 
%Рассмотрим решетку $\varepsilon \mathbb{Z}^2 = \{ (x,t): x/\varepsilon, t/\varepsilon \in \mathbb{Z} \}$.
%\textit{Путь шашки} --- конечная последовательность таких целых точек плоскости, что вектор из каждой точки(кроме последней) к следующей равен либо $(\varepsilon,\varepsilon)$, либо $(-\varepsilon,\varepsilon)$. \textit{Поворот} --- это такая точка пути(не первая и не последняя), что вектор, соединяющий эту точку с предыдущей, ортогонален вектору, соединяющему её со следующей. Для любых $(x,t) \in \mathbb{Z}$, где $t > 0$ обозначим через
 %$$
 %a(x,t,m,\varepsilon) := (1+m^2\varepsilon^2)^{(1-t/\varepsilon)/2}i\sum_{s}(-im\varepsilon)^{t(s)}
 %$$
%сумму по всем путям $s$ шашки на $\varepsilon\mathbb{Z}$ из клетки $(0,0)$ в клетку $(x,t)$, начинающихся с хода в $(\varepsilon, \varepsilon)$.
%\begin{align*}
%    P(x,t,m,\varepsilon) &:= |a(x,t,m,\varepsilon)|^2, \      a_1(x,t,m,\varepsilon):=\mathrm{Re}\, a(x,t,m,\varepsilon), \\       a_2(x,t,m,\varepsilon)&:=\mathrm{Im}\,a(x,t,m,\varepsilon).
%\end{align*}

In this model, the external electromagnetic field is not added artificially, but appears naturally.

The vector $a(s)$ did not rotate, while the checker moved straightly. It goes almost without saying to rotate the vector during the motion. It does not change the model essentially: since all the checker paths from the initial position to the final one have the same length, it follows that all the vectors are rotated by the same angle, which does not affect the probabilities. A more interesting modification arises when the rotation angle depends on the current position of the checker.
This is exactly what the electromagnetic field does. Further, for simplicity, the rotation angle assumes only two values, $0^\circ$ and $180^\circ$, which means multiplication by $\pm 1$.

Thus we understand an electromagnetic field as a fixed arrangement $u$ of numbers $\pm 1$ in the vertices of the squares (see Fig.~\ref{C}). Let us modify the definition of the vector $a(s)$. Now, it changes the direction to the opposite, whenever the paths passes through the vertex with the electromagnetic field $-1$. Denote by ${a}(s,u)$ the resulting vector. We define ${a}(x,t,m,\varepsilon,u)$ and $P(x,t,m,\varepsilon,u)$ analogously to ${a}(x,t,m,\varepsilon)$ and $P(x,t,m,\varepsilon)$, changing ${a}(s)$ to ${a}(s,u)$. For example, if $u$ is identically $+1$, then $P(x,t,m,\varepsilon,u)=P(x,t,m,\varepsilon)$. Figure~\ref{C} depicts another field $u_1$ and the arrows $a(s,u_1)=(0,1/2)$, $a(s_0,u_1)=(1/2,0)$.

We summarize this construction with the following precise definition. \newline
\begin{definition}
\label{def1}
Fix $\varepsilon > 0$ and $m\ge 0$ called \textit{the lattice step} and \textit{the electron mass} respectively. 
Consider the lattice $\varepsilon \mathbb{Z}^2 = \{ (x,t): x/\varepsilon, t/\varepsilon \in \mathbb{Z} \}$. \textit{A checker path} is finite sequence of points of the lattice such that the vector from each point (except the last one) to the next one equals either $(\varepsilon,\varepsilon)$ or $(-\varepsilon,\varepsilon)$. \textit{A turn} is a point of the path (not the first and not the last one) such that the vector from this point to the previous one is orthogonal to the vector from this point to the next one. Denote by $\mathrm{turns}(s)$ the number of turns in the path $s$. \emph{An auxiliary edge} is the segment connecting two nearest points $(x_1,t_1)$ and $(x_2,t_2)$ of the lattice $\varepsilon\mathbb{Z}^2$ such that $(x_1 +t_1)/\varepsilon$ and $(x_2 + t_2)/\varepsilon$ are even. Let $u$ be a map from the set of all auxiliary edges into $\{+1,-1\}$. Denote by
\begin{align*}
&a(x,t,m,\varepsilon,u):= \\
&=(1+m^2\varepsilon^2)^{(1-t/\varepsilon)/2}\,i\,\sum_s (-im\varepsilon)^{\mathrm{turns}(s)}u(s_0s_1)u(s_1s_2)\dots u(s_{t/\varepsilon-1}s_{t/\varepsilon})
\end{align*}
the sum over all checker paths $s=(s_0,s_1,\dots,s_{t/\varepsilon})$ such that $s_0=(0,0)$, $s_1=(\varepsilon,\varepsilon)$,  $s_{t/\varepsilon}=(x,t)$. Hereafter an empty sum is set to be zero by definition.

Denote
$$P(x,t,m,\varepsilon,u) := |a(x,t,m,\varepsilon,u)|^2,$$
$$a_1(x,t,m,\varepsilon,u):=\mathrm{Re}\, a(x,t,m,\varepsilon,u),$$
$$a_2(x,t,m,\varepsilon,u):=\mathrm{Im}\,a(x,t,m,\varepsilon,u).$$
The value $P(x,t,m,\varepsilon,u)$ is called \textit{the probability to find an electron of mass $m$ at the point $(x,t)$ (or, in other words, at the point $x$ at the moment $t$)  on the lattice of step $\varepsilon$, if it was emitted from the point $(0,0)$ and moved in the field $u$}. The value $\sum\limits_{x \in \varepsilon\mathbb{Z}}a_1(x,t,m,\varepsilon,u)^2$ is called \emph{the probability to find the electron with the reversed chirality at the moment $t$}. (The meaning of this terminology is clarified in \cite[\S4]{article}.)
For half-integer $x/\varepsilon,t/\varepsilon$ denote by $u(x,t)$ the value of the field $u$ on the auxiliary edge with the midpoint at $(x,t)$.

The field $u_{\varepsilon}$, given by the formula $$u_{\varepsilon}(x+\varepsilon/2,t+\varepsilon/2)=
\begin{cases}
    -1, &\text{if } (t-x)/4\varepsilon \in \mathbb{Z}, \\
    1, &\text{othewise}.
\end{cases}$$
is called \textit{the homogeneous electromagnetic field} (see Fig.~\ref{C}).
\end{definition}

\begin{figure}[!ht]
%\vspace{-0.5cm}
\begin{center}
\begin{tabular}{c}
\includegraphics[width=.65\linewidth]{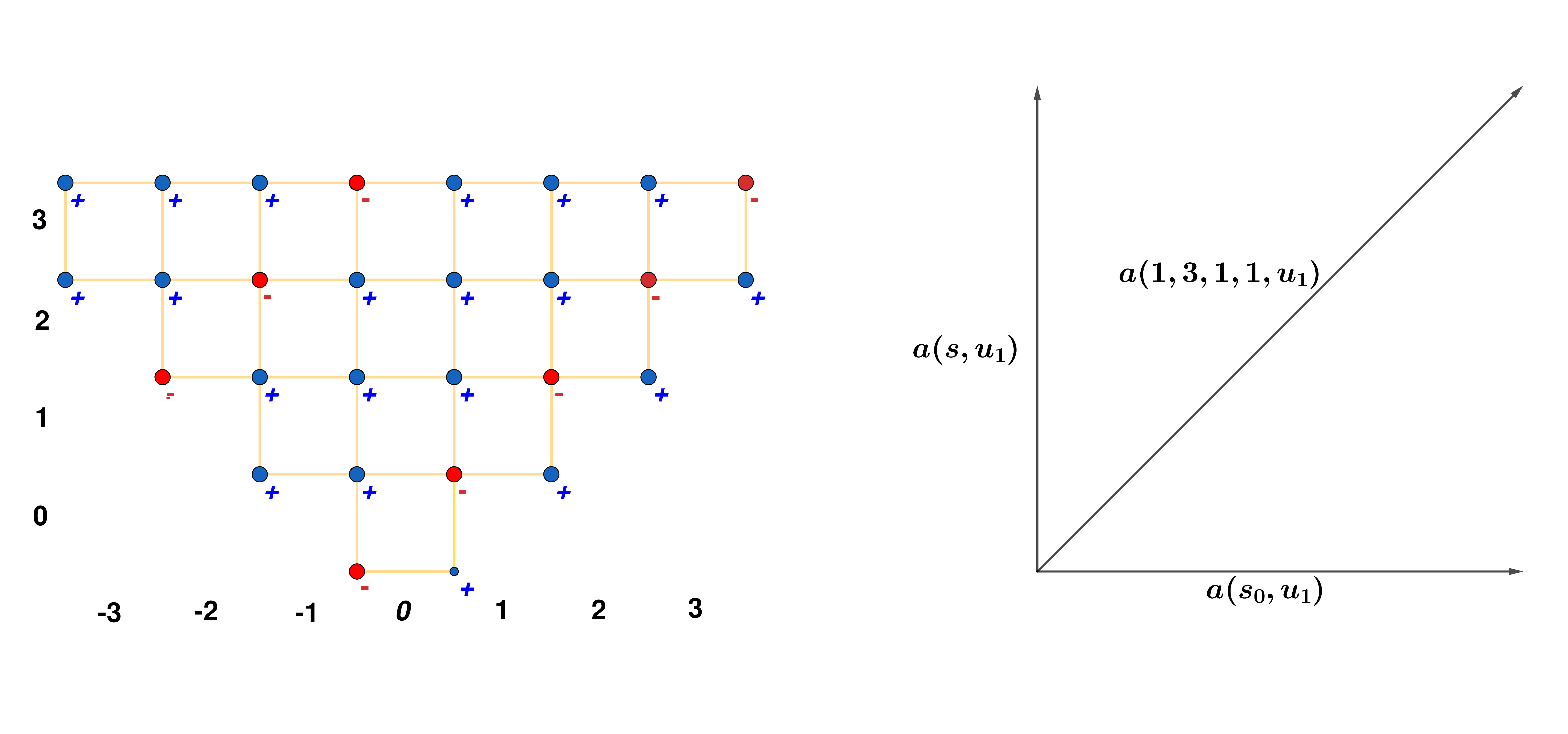}
%\hspace{0.5cm}
\end{tabular}
\end{center}
\vspace{-0.6cm}
\caption{\footnotesize{{The homogeneous electromagnetic field $u_{1}$ (left). Vectors, assigned to the checker paths from Fig.~\ref{Checker-paths} in this field (right).}}}
\label{C}
\vspace{-0.3cm}
\end{figure}

\begin{table}[!h]
\resizebox{\textwidth}{!}{\begin{tabular}{|c|c|c|c|c|c|c|c|}
\hline
$4\varepsilon$&$-\frac{m \varepsilon}{(1+m^2\varepsilon^2)^{3/2}}$&&$\frac{(m \varepsilon+ m^3 \varepsilon^3) \bluevar{- m^2 \varepsilon^2 i}}{(1+m^2\varepsilon^2)^{3/2}}$&&$-\frac{m \varepsilon }{(1+m^2\varepsilon^2)^{3/2}}$&&$\frac{1}{(1+m^2\varepsilon^2)^{3/2}} i$\\
\hline
$3\varepsilon$&&$-\frac{m \varepsilon}{1+m^2\varepsilon^2}$&&$\frac{m \varepsilon + m^2 \varepsilon^2 i}{1+m^2\varepsilon^2}$&&$-\frac{1}{(1+m^2\varepsilon^2)} i$&\\
\hline
$2\varepsilon$&&&$-\frac{m\varepsilon}{\sqrt{1+m^2\varepsilon^2}}$&&$\frac{1}{\sqrt{1+m^2\varepsilon^2}} i$&&\\
\hline
$\varepsilon$&&&&$-i$&&&\\
\hline
%0&&&$0$&&&&\\
%\hline
\diagbox[dir=SW,height=21pt]{$t$}{$x$}&$-2\varepsilon$&$-\varepsilon$&$0$&$\varepsilon$&$2\varepsilon$&$3\varepsilon$&$4\varepsilon$ \\
\hline
\end{tabular}}
\caption{\footnotesize The values of $a(x,t,m,\varepsilon, u_{\varepsilon})$ in the homogeneous electromagnetic field for small $x$ and $t$.}
\label{table-amu}
\end{table}

\par
\begin{remark}
\label{Ueq}
The equality $a(\varepsilon x, \varepsilon t, m,\varepsilon,u_{\varepsilon}) = a(x,t,m\varepsilon,1,u_{1})$ holds for all integer $x,t$.\remove
\end{remark} 

\begin{remark}({Cf. \cite[Remark 5]{article}})
The field $u$ is a fixed classical \bluer{background} field, the electron does not affect it.

This definition is equivalent to one of the first constructions of gauge theory by Weil--Fock--London, and provides a coupling of Feynman checkers to the Wegner--Wilson $\mathbb{Z}/2\mathbb{Z}$ lattice gauge theory.

For instance, the field $u_{\varepsilon}$ in Fig.~\ref{C} has the form $$u_{\varepsilon}(s_1s_2)= \exp\left(-i\int\limits_{s_1s_2}\left(A_0 \mathrm{d} t + A_1 \mathrm{d} x\right) \right)$$ for each auxiliary edge $s_1s_2$, where \bluer{$$(A_0, A_1) := \frac{\pi}{4\varepsilon^2}\left(t-x + 2\varepsilon, t-x + 2\varepsilon \right)$$} is the vector-potential of a constant homogeneous electromagnetic field. This field blows up as $\varepsilon \to 0$, thus the limit in Theorem \ref{fproblim} fails to be interpreted as the wave function of the electron in some continuous external field. Cf. \cite[Problem 12]{article}.
\end{remark}
Table~\ref{table-amu} depicts the values of the wave function of an electron in the model with the homogeneous electromagnetic field for small $x,t$.

\section{Statement of the results}
\label{results}
Let us state our main result (Theorem~\ref{fproblim} above) in a stronger form.
%\begin{wrapfigure}{c}{6.2cm}
%\vspace{-0.4cm}
%\includegraphics[width=6cm]{7v3.png}
%\vspace{-0.2cm}
%\caption{(из ~\cite{article})Пути в поле}
%\label{homogeneous-field}
%\vspace{-1cm}
%\includegraphics[width=6cm]{8a.png}
%\end{wrapfigure}

\begin{theorem}
\label{wavelimit}
Let $u_\varepsilon$ be the homogeneous electromagnetic field. Then for each $m>0$ and $|x|<t$ we have:
\begin{align*}
&\lim\limits_{\varepsilon \searrow 0}\frac{1}{2\varepsilon} a_1\left(4\varepsilon\left\lfloor{\frac{ x}{4\varepsilon}}\right\rfloor,4\varepsilon\left\lfloor{\frac{ t}{4\varepsilon}}\right\rfloor,m,\varepsilon,u_{\varepsilon}\right) =\frac{m}{2} J_0\left(m \sqrt{\frac{t^2-x^2}{2}}\right),\\
&\lim\limits_{\varepsilon \searrow 0}\frac{1}{2\varepsilon} a_2\left(4\varepsilon\left\lfloor{\frac{ x}{4\varepsilon}}\right\rfloor,4\varepsilon\left\lfloor{\frac{ t}{4\varepsilon}}\right\rfloor,m,\varepsilon,u_{\varepsilon}\right) =-\frac{m}{\sqrt{2}} \sqrt{\frac{t+x}{t-x}}J_1\left(m \sqrt{\frac{t^2-x^2}{2}}\right).
\end{align*}
\end{theorem} 

In Figure~\ref{Limit}\bluer{(left)}, dashes depict the graph of the normalized probability $\frac{1}{4\varepsilon^2}P\left(4\varepsilon\left\lfloor{\frac{ x}{4\varepsilon}}\right\rfloor,4\varepsilon\left\lfloor{\frac{ t}{4\varepsilon}}\right\rfloor,m,\varepsilon,u_{\varepsilon}\right)$ for $t = 50, m = 1, \varepsilon = 0.1$ and arbitrary $|x| < t$, which is a multiple of $0.1$, while the curve represents the expression with Bessel functions from Theorem~\ref{fproblim}, i.e. the pointwise limit of the normalized probability at $\varepsilon \to 0$.

In order to prove Theorem~\ref{wavelimit} we need the following proposition.

Denote by $\delta_2(b)$ the remainder of an integer $b$ after the division by 2. We also set ${a\choose j}=0$, if integer $a < 0$ or $j < 0$.

\begin{proposition}({Cf. \cite[Proposition 11]{article}})
\label{formula}
For each real $m \ge 0$ and integer $\xi,\eta \ge 0$ the following equalities hold:
\begin{align*}
a_1(\xi - \eta + 1,&\xi + \eta + 1,m,1,u_{1})=\\
&=(-1)^{\xi+1}\frac{m(1+m^2)^{\delta_2(\xi(\eta+1))} }{(1+m^2)^{\frac{\xi+\eta}{2}}}\sum_{j=0}^{\lfloor\frac{\xi}{2}\rfloor}{\lfloor\frac{\xi}{2}\rfloor\choose j}{\lfloor\frac{\eta-1}{2}\rfloor\choose j}(1 - (1+m^2)^2)^j ;\\
a_2(\xi - \eta + 1,&\xi + \eta+1,m,1,u_{1})=\\
&=\frac{(-1)^{\xi+1}}{(1+m^2)^{\frac{\xi+\eta}{2}}}\sum_{j=0}^{\lfloor\frac{\xi}{2}\rfloor}\left( {\lfloor\frac{\eta}{2}\rfloor\choose j}(1 +m^2)^{\delta_2(\xi\eta)}-\right.\\ 
&-\left.{\lfloor\frac{\eta-1}{2}\rfloor\choose j}(1+m^2)^{\delta_2(\xi(\eta+1))} \right){\lfloor\frac{\xi}{2}\rfloor\choose j}(1 - (1+m^2)^2)^j.
\end{align*}
\end{proposition}
\begin{remark} In particular, $a_2(\xi - \eta + 1,\xi + \eta+1,m,1,u_{1}) =0$ for $\eta - 1\equiv_2 \xi \equiv_2 0$.
\end{remark}

Formulae from Proposition~\ref{formula} can be rewritten in terms of hypergeometric functions. For integer $a,b,c$, where $b \le 0$ and $c>0$, the polynomial
$${}_2F_1\left(a,b;c;z\right) := 1 + \sum_{k=1}^{\infty}\prod_{l=0}^{k-1}\frac{(a+l)(b+l)}{(1+l)(c+l)}z^k$$
is called \emph{a hypergeometric function} (or \emph{a Jacobi polynomial}). Note that for integer $b \le 0$ this sum has a finite number of nonzero terms.

\begin{proposition}({Cf. \cite[Remark 3]{article}})
\label{hyper-form}
Denote $z := 1 - (1+m^2)^2$. For each real $m \ge 0$ and integer $\xi, \eta \ge 0$ we have
\footnotesize
\remove{}
\begin{align*}
&a_1(\xi-\eta+1,\xi+\eta+1,m,1,u_1)=\\
&=(-1)^{ \xi+1}m(1+m^2)^{-\frac{\xi+\eta}{2}+\delta_2((1+\eta)\xi)}\cdot{}_2F_1\left( - \left\lfloor\frac{\eta-1}{2}\right\rfloor  ,-\left\lfloor\frac{\xi}{2}\right\rfloor;1;z\right);\\
&a_2(\xi-\eta+1,\xi+\eta+1,m,1,u_1)=\\
&=\begin{dcases}
\frac{-\xi}{2}(1+m^2)^{-\frac{\xi+\eta}{2}}z\cdot{}_2F_1\left( - \frac{\eta}{2} + 1 ,-\frac{\xi}{2}+ 1;2;z\right), &\text{if } \xi \equiv_2 \eta \equiv_2 0,\\
\frac{\bluevar{(\xi-1)} z\cdot{}_2F_1\left( - \frac{\eta}{2} + 1 ,-\frac{\xi-1}{2} + 1;2;z\right) \bluevar{-} 2m^2\cdot{}_2F_1\left( - \frac{\eta}{2}+1,-\frac{\xi-1}{2};1;z\right)}{2(1+m^2)^{\frac{\xi+\eta}{2}}}, &\text{if } \xi - 1 \equiv_2 \eta \equiv_2 0,\\
0 , &\text{if } \xi  \equiv_2 \eta - 1 \equiv_2 0,\\
(1+m^2)^{-\frac{\xi+\eta}{2}}m^2\cdot{}_2F_1\left( - \frac{\eta-1}{2},-\frac{\xi-1}{2};1;z\right), &\text{if } \xi \equiv_2\eta \equiv_2 1.
\end{dcases}
\end{align*}
\normalsize
\end{proposition}
 
We provide two proofs of Proposition~\ref{formula}. The first one (see \S \ref{second-proof}) uses the method of generating functions (cf. ~\cite[Appendix A]{article}) and the machinary of hypergeometric functions, while the second one (see \ref{first-proof}) is elementary, but does not clarify where the formulae come from. 

Let us state Theorems~\ref{wl-intr}~and~\ref{chir-rev-f} in a stronger form.

\begin{theorem}({Cf.~\cite[Theorem 1]{article}})
\label{large-time-lim} For each real $m, \varepsilon> 0$ we have:
\begin{enumerate}
    \item[(A)] For each real $v$ the following equality holds  $$\lim_{\substack{t \to \infty \\ t \in \varepsilon\mathbb{Z}}}\sum_{\substack{x\le v t\\ x \in \varepsilon\mathbb{Z}}}P(x,t,m,\varepsilon,u_{\varepsilon})= F(v):=
    \begin{cases}
        0 ,& \text{if } v < -\frac{1}{1+m^2\varepsilon^2};\\
        \frac{1}{\pi}\arccos{\frac{1-(1+m^2\varepsilon^2)^2v}{(1+m^2\varepsilon^2)(1-v)}}, &\text{if } |v| \le \frac{1}{1+m^2\varepsilon^2};\\
        1 , &\text{if } v > \frac{1}{1+m^2\varepsilon^2}.
    \end{cases} $$
    \item[(B)] For each real $v$ there is the following convergence in 
    distribution as $t\to\infty$, where $t \in \varepsilon\mathbb{Z}$:
    $$
    \frac{t}{\varepsilon} P\left(\left\lceil \frac{vt}{\varepsilon}\right\rceil\varepsilon, t,m,\varepsilon,u_{\varepsilon}\right) \overset{d}{\to} F^{\prime}(v) = 
    \begin{cases}
        \frac{\sqrt{(1+m^2\varepsilon^2)^2 - 1}}{\pi(1-v)\sqrt{1-(1+m^2\varepsilon^2)^2v^2}}, & \text{if } |v| \le \frac{1}{1+m^2\varepsilon^2};\\
        0 , & \text{if } |v| > \frac{1}{1+m^2\varepsilon^2}.
    \end{cases}
    $$
    \item[(C)] For each integer $r \ge 0$ we have
    \\
    $\lim\limits_{\substack{t \to \infty \\ t \in \varepsilon\mathbb{Z}}}\sum\limits_{x \in \varepsilon\mathbb{Z}}\frac{x^r}{t^r}P(x,t,m,\varepsilon,u_{\varepsilon})=\int\limits_{-1}^{1}v^r F^{\prime}(v) dv.$
\end{enumerate}

\end{theorem}

\begin{theorem}({Cf. \cite[Theorem 2]{IB-20}})
\label{prob-rev-field}
For each $m,\varepsilon, t > 0$, where $t\in \varepsilon\mathbb{Z}$, we have
\begin{align*}
    \sum_{\substack{x \in \varepsilon\mathbb{Z}}}a_1^2(x,t,m,\varepsilon,u_{\varepsilon}) = 
    \begin{cases}
        \frac{m\varepsilon}{(1+m^2\varepsilon^2)\sqrt{2 + m^2\varepsilon^2}} + O_{m,\varepsilon}\left(t^{-1/3}\right) ,& \text{if } \frac{t}{\varepsilon} \equiv_{2} 1;\\
       \frac{m\varepsilon}{\sqrt{2 + m^2\varepsilon^2}} + O_{m,\varepsilon}\left(t^{-1/3}\right),& \text{if } \frac{t}{\varepsilon} \equiv_{2} 0.
    \end{cases} 
\end{align*}
\end{theorem}

For $m=\varepsilon=1$ this result confirms the conjecture stated by I. Bogdanov in \cite[Hypothesis 1]{IB-20}.

\bluer{Hereafter the notation $f(t,m,\varepsilon) = g(t,m,\varepsilon) + O_{m,\varepsilon}(h(t))$ means that there exists a constant $C(m,\varepsilon)$ (depending on $m,\varepsilon$, but independent of $t$) such that for each $t$ satisfying the conditions of the theorem, we have  $|f(t,m,\varepsilon) - g(t,m,\varepsilon)| \le C(m,\varepsilon)h(t)$.}

\bluer{We announce the following uniform asymptotic formula for the wave function. It is based on Theorem \ref{Airy_Zak} proved in \cite{Zakorko}, which is not publicly available, but we expect another proof to appear in \cite{Skopenkov-Ustinov-Kuyanov-Drmota-23}.}
\begin{theorem}
\label{Airy}
    For $x/2\varepsilon$ and $t/2\varepsilon + 1$ even, $0<m\varepsilon < 1$, and $\left|\frac{x}{t}\right| < \frac{1}{1+m^2\varepsilon^2}$ we have
    \begin{align*}
    &a_1(x,t) =(-1)^{\left\lfloor\frac{x+t+\bluevar{4\varepsilon}}{4\varepsilon} \right\rfloor}\frac{\sqrt{2}m^{1/2}\varepsilon (1+m^2\varepsilon^2)^{1/2}\left(-12\tilde{\theta}(x/t)\right)^{1/6}}{((2 + m^2\varepsilon^2)(1-(x/t)^2(1+m^2\varepsilon^2)^2))^{1/4}} \cdot \\
    &\cdot \left(\frac{1}{t}\right)^{1/3} \mathrm{Ai }\left(-\left(-\frac{3}{2}\tilde{\theta} \left(\frac{x}{t}\right)t\right) ^{2/3}\right) +O_{m,\varepsilon}\left(\frac{1}{t}\right),  
    \end{align*}
    \normalsize
    where 
    $$\tilde{\theta} (v) := \frac{1}{2\varepsilon}\left(|v| \arccos \frac{|v|\sqrt{(1+m^2\varepsilon^2)^2-1}}{\sqrt{1-v^2}} - \arccos \frac{\sqrt{(1+m^2\varepsilon^2)^2-1}}{(1+m^2\varepsilon^2)\sqrt{1-v^2}}\right).$$
\end{theorem}
\normalsize

Our method of the proof of Theorem \ref{prob-rev-field} differs from the \bluer{approach} suggested in \cite[\S 12.7]{article} and \cite[\S 3]{IB-20}. In order to prove Theorems~\ref{large-time-lim}, \ref{prob-rev-field}, and \ref{Airy} we apply the method suggested in~\cite[~\S 2]{Grimmet-Janson-Scudo-04}, \cite{Anikin2019, Zakorko}, and the following new result, which will be proved by induction in \ref{IntFormProof}.

\begin{proposition}({Cf.~\cite[Proposition 12]{article}})
\label{IntForm}
Set $$\omega_p:= \frac{1}{2\varepsilon}\arcsin{\frac{\sin{2p\varepsilon}}{1+m^2\varepsilon^2}}.$$ Then
for each $m > 0$  and $(x,t)\in \varepsilon\mathbb{Z}^2$, where $t>0$, we have
\begin{align*}
&a_1(x,t,m,\varepsilon,u_{\varepsilon}) =\\
&=\footnotesize{\begin{cases}
    (-1)^{\frac{x-\varepsilon}{2\varepsilon} + \left\lfloor\frac{x+t}{4\varepsilon}\right\rfloor}\frac{m\varepsilon^2}{\pi(1+m^2\varepsilon^2)}\int\limits_{-\pi/ \varepsilon}^{\pi/ \varepsilon}e^{ipx}\sin{\omega_p(t-\varepsilon)}\cdot\frac{\sin{p\varepsilon}}{\cos{2\omega_p\varepsilon}} dp, &\text{if } \frac{t}{\varepsilon} \equiv_4 1,\\
    (-1)^{\frac{x}{2\varepsilon} + \left\lfloor\frac{x+t}{4\varepsilon}\right\rfloor}\frac{m\varepsilon^2}{2\pi\sqrt{1+m^2\varepsilon^2}}\int\limits_{-\pi/\varepsilon}^{\pi/\varepsilon}e^{ipx}\cdot\frac{i\sin{\omega_p(t-2\varepsilon) - \cos{\omega_pt} }}{\cos{2\omega_p\varepsilon}} dp, &\text{if } \frac{t}{\varepsilon} \equiv_4 2,\\
    (-1)^{\frac{x-\varepsilon}{2\varepsilon} + \left\lfloor\frac{x+t}{4\varepsilon}\right\rfloor}\frac{im\varepsilon^2}{\pi(1+m^2\varepsilon^2)}\int\limits_{-\pi/\varepsilon}^{\pi/\varepsilon}e^{ipx}\cos{\omega_p(t-\varepsilon)}\cdot\frac{\sin{p\varepsilon}}{\cos{2\omega_p\varepsilon}} dp, &\text{if } \frac{t}{\varepsilon} \equiv_4 3,\\
    (-1)^{\frac{x+2\varepsilon}{2\varepsilon}+\left\lfloor\frac{x+t}{4\varepsilon}\right\rfloor}\frac{m\varepsilon^2}{2\pi\sqrt{1+m^2\varepsilon^2}}\int\limits_{-\pi/\varepsilon}^{\pi/\varepsilon}e^{ipx}\frac{\cos{\omega}_p(t-2\varepsilon) - i\sin{\omega}_pt}{\cos{2\omega_p\varepsilon}} dp, &\text{if } \frac{t}{\varepsilon} \equiv_4 0.
\end{cases}}\\
&a_2(x,t,m,\varepsilon,u_{\varepsilon}) = \\
&=\footnotesize{\begin{cases}
    (-1)^{\frac{x-\varepsilon}{2\varepsilon}+\left\lfloor\frac{x+t}{4\varepsilon}\right\rfloor}\frac{\varepsilon}{2\pi}\int\limits_{-\pi/ \varepsilon}^{\pi/ \varepsilon}e^{ip(x-\varepsilon)}\left(i\sin{\omega_p(t-\varepsilon)}\cdot\frac{m^2\varepsilon^2 + \cos{2p\varepsilon}}{(1+m^2\varepsilon^2)\cos{2\omega_p\varepsilon}} - \cos{\omega_p(t-\varepsilon)}\right)dp,\\ \hspace{8cm} \text{if } \frac{t}{\varepsilon} \equiv_4 1,\\
    (-1)^{\frac{x-2\varepsilon}{2\varepsilon}+\left\lfloor\frac{x+t}{4\varepsilon}\right\rfloor}\frac{\varepsilon}{2\pi\sqrt{1+m^2\varepsilon^2}}\int\limits_{-\pi/\varepsilon}^{\pi/\varepsilon}e^{ip(x-2\varepsilon)}\cdot\frac{ie^{2ip\varepsilon}\sin{\omega_p(t-2\varepsilon)}  - \cos{\omega_pt} }{\cos{2\omega_p\varepsilon}} dp, \\ \hspace{8cm} \text{if } \frac{t}{\varepsilon} \equiv_4 2,\\
    (-1)^{\frac{x+\varepsilon}{2\varepsilon}+\left\lfloor\frac{x+t}{4\varepsilon}\right\rfloor}\frac{\varepsilon}{2\pi}\int\limits_{-\pi/\varepsilon}^{\pi/\varepsilon}e^{ip(x-\varepsilon)}\left(\cos{\omega_p(t-\varepsilon)}\cdot\frac{m^2\varepsilon^2 +\cos{2p\varepsilon}}{(1+m^2\varepsilon^2)\cos{2\omega_p\varepsilon}} - i\sin{\omega_p(t-\varepsilon)}\right) dp, \\ \hspace{8cm} \text{if } \frac{t}{\varepsilon} \equiv_4 3,\\
    (-1)^{\frac{x}{2\varepsilon}+\left\lfloor\frac{x+t}{4\varepsilon}\right\rfloor}\frac{\varepsilon}{2\pi\sqrt{1+m^2\varepsilon^2}}\int\limits_{-\pi/\varepsilon}^{\pi/\varepsilon}e^{ip(x-2\varepsilon)}\cdot
    \frac{e^{2ip\varepsilon}\cos{\omega_p(t-2\varepsilon)} - i\sin{\omega_pt}}{\cos{2\omega_p\varepsilon}}
     dp,\\ \hspace{8cm} \text{if } \frac{t}{\varepsilon} \equiv_4 0.
\end{cases}}
\end{align*}
\end{proposition}

Here for definiteness we set $(-1)^{n/2}:=i^n$ for odd $n$ (although this expression occurs only in the combination $(-1)^{n/2}\cdot 0 =0$).

These integrals represent a wave, emitted by a point source as a superposition of waves with wavelength $2\pi/p$ and frequency $\omega_p$.

This proposition can be derived by solving the lattice Dirac equation (Lemma \ref{dirak}) by Fourier method, but we prefer to give a direct cheking by induction in \ref{IntFormProof}.

\begin{remark}
Our expression for $\omega_p$ is equivalent to $\omega_{\pm}(k)$ from \cite[Equation (7)]{Cedzich-et-al-2013} for $m = 4$, $a= 1/\sqrt{1+m^2\varepsilon^2}$ (where is in the left sides we use the notation from \cite{Cedzich-et-al-2013} different from ours).
\end{remark}

%Having Proposition \ref{IntForm} allows us to write the following asymptotic formula for the wave function, using approach by Anikin et. al \cite{Anikin2019, Zakorko}.

In order to prove Propositions~\ref{formula} and \ref{IntForm} we use the following known result. It is proved analogously to~\cite[Proposition 5]{article}.  
\begin{lemma}[Dirac equation in an electromagnetic field on a lattice] (Cf.~\cite[Propositions 5 and 14]{article})
\label{dirak}
For each \bluer{$\xi, \eta \in \varepsilon\mathbb{Z}$ such that $\xi, \eta >0$} we have
\footnotesize{\begin{align*}
&a_1(\xi - \eta + \varepsilon,\xi + \eta + \varepsilon,m,\varepsilon,u_{\varepsilon}) =\frac{1}{\sqrt{1+m^2\varepsilon^2}} u_{\varepsilon}\left(\xi - \eta +\frac{3\varepsilon}{2}, \xi + \eta +\frac{\varepsilon}{2}\right)\cdot\\
&\cdot(a_1(\xi - \eta +2\varepsilon, \xi +\eta , m,\varepsilon,u_{\varepsilon}) + m\varepsilon a_2(\xi -\eta +2\varepsilon, \xi + \eta,m,\varepsilon,u_{\varepsilon})),\\
&a_2(\xi - \eta + \varepsilon,\xi + \eta + \varepsilon,m,\varepsilon,u_{\varepsilon}) =\frac{1}{\sqrt{1+m^2\varepsilon^2}} u_{\varepsilon}\left(\xi - \eta +\frac{\varepsilon}{2}, \xi + \eta +\frac{\varepsilon}{2}\right)\cdot\\
&\cdot(-m\varepsilon a_1(\xi - \eta ,\xi + \eta,m,\varepsilon,u_{\varepsilon}) + a_2(\xi - \eta,\xi + \eta,m,\varepsilon,u_{\varepsilon})).
\end{align*}}
\end{lemma}

\normalsize
\section{First proof of Proposition~\ref{formula}}
\label{second-proof}

In this section we use the following notation:
\begin{align*}
    \hat{A}_1(p,q) := \sum_{\xi,\eta \ge 0}(1+m^2)^{\frac{\xi+\eta}{2}}a_1(\xi-\eta + 1, \xi + \eta + 1,m,1,u_1)p^{\xi}q^{\eta},\\
    \hat{A}_2(p,q) := \sum_{\xi,\eta \ge 0} (1+m^2)^{\frac{\xi + \eta}{2}}a_2(\xi-\eta + 1, \xi +\eta + 1,m,1,u_1)p^{\xi}q^{\eta} .
\end{align*}
\normalsize
\begin{lemma}
\label{generSum}
We have the following equalities of formal power series:

$$ \hat{A}_1(p,q) = -\frac{m q \left(1 - p + q - pq\left(m^2+1\right)\right)}{1-\left(q^2 + p^2 -\left(m^2+1\right)^2p^2q^2\right) },$$

$$
 \hat{A}_2(p,q) =  \frac{(1-q) \left(1 - p + q - pq\left(m^2+1\right)\right)}{1-\left(q^2 + p^2 -\left(m^2+1\right)^2p^2q^2\right) }.
$$
\end{lemma}
\begin{proof}\textbf{({\bluer{Cf. \cite[Proposition 18]{article}}})}
For $g=0$ or $1$ and $i=1$ or $2$ denote
$$B^{g}_i(p,q):= \sum\limits_{\xi, \eta \ge 0: \eta \equiv_{2} g}(1+m^2)^{\frac{\xi + \eta}{2}}a_i(\xi - \eta + 1, \xi + \eta +1 ,m,u_1)p^{\xi}q^{\eta}, $$
Then, by Dirac's equation (Lemma~\ref{dirak}) we have
$$
\begin{pmatrix}
-1& 0& q& mq\\
pm&-1-p&0&0\\
q& mq&-1&0\\
0&0&-pm& -1+p\\
\end{pmatrix}
\begin{pmatrix}
B^0_1(p,q)\\
B^0_2(p,q)\\
B^1_1(p,q)\\
B^1_2(p,q)\\
\end{pmatrix}
=-
\begin{pmatrix}
B^0_1(p,0)\\
B^0_2(0,q)\\
B^1_1(p,0)\\
B^1_2(0,q)\\
\end{pmatrix}=
\begin{pmatrix}
0\\
1\\
0\\
0\\
\end{pmatrix}.
$$

Solving this system and expressing $ \hat{A}_i(p,q)$ through $B_i^{j}$, where $j=0,1$, we obtain the lemma.
\end{proof}

\begin{lemma}For each integer $\xi,\eta \ge 0$ we have
\label{firstclosedf}
\begin{align*}
&(1+m^2)^{\frac{\xi+\eta}{2}}a_1(\xi - \eta +1, \xi + \eta +1,m,1,u_1) =\\
&=m(1 + m^2)^{\xi + \eta-2 + \delta_2((1+\xi)\eta)}(-1)^{\lfloor\frac{3\xi}{2}\rfloor + \lfloor\frac{\eta+1}{2}\rfloor}\cdot \\
&\cdot \sum\limits_{l=\max\left(\lfloor\frac{\xi}{2}\rfloor,\lfloor\frac{\eta-1}{2}\rfloor\right)}^{\lfloor\frac{\xi}{2}\rfloor + \lfloor\frac{\eta-1}{2}\rfloor}{l\choose \lfloor{\frac{\xi}{2}}\rfloor}{\lfloor{\frac{\xi}{2}}\rfloor\choose \lfloor\frac{\xi}{2}\rfloor + \lfloor\frac{\eta-1}{2}\rfloor - l}(1 + m^2)^{-2l}(-1)^{l}.
\end{align*}
\end{lemma}
\begin{proof} 
By Lemma~\ref{generSum} we have
$$
 \hat{A}_1(p,q) = -mq\left(1 - p + q - pq\left(m^2+1\right)\right)\sum\limits_{l=0}^{\infty}\left(p^2 + q^2 - p^2q^2(1+m^2)^2 \right)^l.
$$
Now we compute the coefficient at $p^{\xi}q^{\eta}$. \bluer{Depending} on the choice of $1, p, q, (1+m^2)pq$ from the factor $(1 - p+q -(1+m^2)pq)$, we are left to compute the coefficient at $p^{\xi}q^{\eta-1},p^{\xi-1}q^{\eta-1},$ $p^{\xi}q^{\eta-2},p^{\xi-1}q^{\eta-2}$  in the remaining sum respectively. Note that in the remaining sum there are only monomials with even powers. Then there are 4 cases:

\begin{align*}
&\frac{(1+m^2)^{\frac{\xi+\eta}{2}}}{m}a_1(\xi - \eta + 1, \xi + \eta + 1, m,1,u_1) =\\
&=\begin{cases}
-\sum\limits_{l=\max(\frac{\xi}{2}, \frac{\eta-2}{2})}^{\frac{\xi+\eta-2}{2}}{l\choose\frac{\xi}{2}}{\frac{\xi}{2}\choose\frac{\xi + \eta -2}{2}-l}\left(-(1+m^2)^2\right)^{\frac{\xi + \eta -2}{2}-l},\\
\hspace{8cm} \text{if } \xi \equiv_2 \eta \equiv_2 0, \\
(1+m^2)\sum\limits_{l=\max (\frac{\xi-1}{2}, \frac{\eta-2}{2})}^{\frac{\xi+\eta -3}{2}}{l \choose\frac{\xi-1}{2}}{\frac{\xi-1}{2}\choose\frac{\xi + \eta -3}{2} - l}\left(-(1+m^2)^2\right)^{\frac{\xi+\eta-3}{2}-l},\\
\hspace{8cm}\text{if } \xi-1\equiv_2\eta\equiv_2 0,  \\
-\sum\limits_{l=\max(\frac{\xi}{2}, \frac{\eta-1}{2})}^{\frac{\xi+\eta-1}{2}}{l \choose \frac{\xi}{2}}{\frac{\xi}{2}\choose\frac{\xi+\eta-1}{2}-l}\left(-(1+m^2)^2\right)^{\frac{\xi+\eta-1}{2}-l},\\
\hspace{8cm}\text{if } \xi-1 \equiv_2 \eta\equiv_2 1, \\
\sum\limits_{l=\max(\frac{\xi-1}{2}, \frac{\eta-1}{2})}^{\frac{\xi+\eta-2}{2}}{l\choose\frac{\xi-1}{2}}{\frac{\xi-1}{2}\choose\frac{\xi+\eta-2}{2}-l}\left(-(1+m^2)^2\right)^{\frac{\xi+\eta-2}{2}-l}, \\
\hspace{8cm}\text{if }\xi\equiv_2\eta\equiv_2 1. \\
\end{cases}
\end{align*}
\end{proof}

Now, to prove the Proposition~\ref{formula} we use hypergeometric functions and some known identities for them.
Function ${}_2F_1\left(a,b;c;z\right)$ for integer $c<b\le 0$ is defined analogously to the case $c > 0$ (see~\S\ref{results}), but the summation is from $k=0$ till $k=|b|$ only.

\begin{lemma}
\label{hyper}
For each $m \ge 0$ and integer $\xi,\eta \ge 0$ we have
\begin{align*}
&a_1(\xi - \eta +1, \xi+\eta+1,m,1,u_1) =\\
&=(-1)^{\xi+1}\frac{m(1+m^2)^{\delta_2((1+\eta)\xi)}}{(1+m^2)^{\frac{\xi+\eta}{2}}}{\lfloor\frac{\xi}{2}\rfloor + \lfloor\frac{\eta-1}{2}\rfloor \choose \lfloor \frac{\xi}{2}\rfloor} \cdot\\
&\cdot{}_2F_1\left(- \left\lfloor\frac{\eta-1}{2}\right\rfloor  ,-\left\lfloor\frac{\xi}{2}\right\rfloor; -\left\lfloor\frac{\xi}{2}\right\rfloor - \left\lfloor\frac{\eta-1}{2}\right\rfloor ;(1+m^2)^2\right).
\end{align*}
\end{lemma}
\begin{proof}
Let us transform the expression from Lemma~\ref{firstclosedf}.
Note that $$\xi + \eta - 2 + \delta_2((1+\xi)\eta) = 2\left(\left\lfloor\frac{\xi}{2}\right\rfloor + \left\lfloor\frac{\eta-1}{2}\right\rfloor\right) + \delta_2((1+\eta)\xi)$$ and change the summation index ($l$ to $\lfloor\frac{\xi}{2}\rfloor + \lfloor\frac{\eta-1}{2}\rfloor - l = j$). We obtain

\begin{align*}
&(1+m^2)^{\frac{\xi+\eta}{2}-\delta_2((1+\eta)\xi)}a_1(\xi-\eta+1,\xi+\eta+1,m,1,u_1)=\\
&=(-1)^{\xi+1}m\sum_{j=0}^{\min(\lfloor\frac{\eta-1}{2}\rfloor, \lfloor\frac{\xi}{2}\rfloor)}\frac{(\lfloor\frac{\xi}{2}\rfloor + \lfloor\frac{\eta-1}{2}\rfloor - j)!}{j!(\lfloor\frac{\eta-1}{2}\rfloor -j)!(\lfloor\frac{\xi}{2}\rfloor-j)!}(1 + m^2)^{2j}(-1)^{-j}.
\end{align*}
\normalsize
 Now we bring this series to the hypergeometric form. The ratio of the $(j+1)$-th and the $j$-th term in the obtained sum is equal to
 $$- \frac{( \lfloor\frac{\eta-1}{2}\rfloor  -j)(\lfloor\frac{\xi}{2}\rfloor-j)}{(\lfloor\frac{\xi}{2}\rfloor + \lfloor\frac{\eta-1}{2}\rfloor-j)(j+1)}(1+m^2)^2.
 $$
This is the same as the ratio of the $(j+1)$-th and the $j$-th term for the hypergeometric function with $a = - \lfloor\frac{\eta-1}{2}\rfloor$,  $b = -\lfloor\frac{\xi}{2}\rfloor$, $c = -\lfloor\frac{\xi}{2}\rfloor - \lfloor\frac{\eta-1}{2}\rfloor$. 
Thus we have
 \begin{align*}
 &(1+m^2)^{\frac{\xi+\eta}{2}-\delta_2((1+\eta)\xi)}a_1(\xi-\eta+1,\xi+\eta+1,m,1,u_1)=(-1)^{\xi+1}m\binom{\left\lfloor\frac{\xi}{2}\right\rfloor + \left\lfloor\frac{\eta-1}{2}\right\rfloor}{\left\lfloor\frac{\xi}{2}\right\rfloor}\cdot\\
 &\cdot {}_2F_1\left(- \left\lfloor\frac{\eta-1}{2}\right\rfloor ,-\left\lfloor\frac{\xi}{2}\right\rfloor;-\left\lfloor\frac{\xi}{2}\right\rfloor - \left\lfloor\frac{\eta-1}{2}\right\rfloor;(1+m^2)^2\right).
 \end{align*}
\end{proof}

\textbf{Proof of Propositions~\ref{formula} and \ref{hyper-form}.}
Cases $\eta=0,1,2$ are checked by the direct substitution of $\eta$ in Lemma~\ref{firstclosedf}, hence we consider $\eta \ge 3$.
By the known identity $${}_2F_1\left( a,-n, c;z\right) =\left(\prod\limits_{l=0}^{n-1}\frac{l - c+ a}{l + c}\right){}_2F_1\left( a,-n, 1 - n + a - c; 1- z\right)(see~\cite[(5.105)]{math})$$  
we get
\begin{multline*}
{}_2F_1\left(- \left\lfloor\frac{\eta-1}{2}\right\rfloor ,-\left\lfloor\frac{\xi}{2}\right\rfloor;-\left\lfloor\frac{\xi}{2}\right\rfloor - \left\lfloor\frac{\eta-1}{2}\right\rfloor;(1+m^2)^2\right)=\\
=\left(\prod\limits_{l=0}^{\lfloor\frac{\xi}{2}\rfloor-1}\frac{l-\lfloor\frac{\xi}{2}\rfloor}{l-\lfloor\frac{\xi}{2}\rfloor - \lfloor\frac{\eta-1}{2}\rfloor }\right)\cdot{}_2F_1\left(- \left\lfloor\frac{\eta-1}{2}\right\rfloor ,-\left\lfloor\frac{\xi}{2}\right\rfloor;1;1-(1+m^2)^2\right).
\end{multline*}
Since $$\prod\limits_{l=0}^{\lfloor\frac{\xi}{2}\rfloor}\frac{l-\lfloor\frac{\xi}{2}\rfloor}{l-\lfloor\frac{\xi}{2}\rfloor - \lfloor\frac{\eta}{2}\rfloor} = \frac{1}{{\lfloor\frac{\xi}{2}\rfloor + \lfloor\frac{\eta-1}{2}\rfloor\choose \lfloor \frac{\xi}{2}\rfloor}},$$  

by Lemma~\ref{hyper} we have

\begin{align*}
&a_1(\xi-\eta+1,\xi+\eta+1,m,1,u_1)= \\
&=\frac{(-1)^{ \xi+1}m}{(1+m^2)^{\frac{\xi+\eta}{2}-\delta_2((1+\eta)\xi)}}\cdot{}_2F_1\left( - \left\lfloor\frac{\eta-1}{2}\right\rfloor ,-\left\lfloor\frac{\xi}{2}\right\rfloor;1;1-(1+m^2)^2\right).
\end{align*}

This proves the formula for $a_1(\xi-\eta+1,\xi+\eta+1,m,1,u_1)$ from Proposition~\ref{hyper-form}. Now, representing the hypergeometric function as series with a finite number of non-zero terms, we obtain the formula from Proposition~\ref{formula} for $\eta \ge 3$.

The formula for $a_2(\xi-\eta+1,\xi+\eta+1,m,1,u_1)$ is obtained from the one for $a_1(\xi-\eta+1,\xi+\eta+1,m,1,u_1)$ using Dirac's equation (Lemma~\ref{dirak}). $\square$

\section{Proofs of Theorems~\ref{fproblim} and~\ref{wavelimit}}
In this section we use the following notation: $A:=\left\lfloor{\frac{x}{4\varepsilon}}\right\rfloor + \left\lfloor{\frac{ t}{4\varepsilon}}\right\rfloor,$ $B:= \left\lfloor{\frac{t}{4\varepsilon}}\right\rfloor - \left\lfloor{\frac{ x}{4\varepsilon}}\right\rfloor$.
 
\smallskip \textbf{The proof of Theorem~\ref{wavelimit} modulo Lemmas~\ref{norm} --~\ref{binom}.}
The first equality in the theorem follows from chain of formulae
\begin{align*}
&\frac{1}{2\varepsilon}a_1\left(4\varepsilon\left\lfloor\frac{x}{4\varepsilon}\right\rfloor, 4\varepsilon\left\lfloor\frac{t}{4\varepsilon}\right\rfloor, m, \varepsilon, u_{\varepsilon}\right) = \frac{1}{2\varepsilon}a_1\left(2(A-B), 2(A+B), m\varepsilon, 1, u_1\right) =  \\
&=\frac{1}{2\varepsilon}m\varepsilon\left(1+(m\varepsilon)^2\right)^{3/2-2\left\lfloor\frac{ t}{4\varepsilon}\right\rfloor} \sum_{j=0}^{A-1}(-1)^j{A - 1\choose j}{B - 1\choose j}\left(2 + (m\varepsilon)^2\right)^j(m\varepsilon)^{2j} \sim \\
&\sim\frac{m}{2}\sum_{j=0}^{\infty}(-1)^j{A - 1\choose j}{B - 1\choose j}\left(2 + (m\varepsilon)^2\right)^j(m\varepsilon)^{2j} =\\
&=\frac{m}{2}\sum_{2|j}{A - 1\choose j}{B - 1\choose j}\left(2 + (m\varepsilon)^2\right)^j(m\varepsilon)^{2j} -\\
&-\frac{m}{2}\sum_{2 \nmid j}{A - 1\choose j}{B - 1\choose j}\left(2 + (m\varepsilon)^2\right)^j(m\varepsilon)^{2j} \to  \\
&\to \frac{m}{2}\left(\sum_{2 |j}\frac{(x+t)^j(t-x)^j}{8^j(j!)^2}m^{2j} - \sum_{2 \nmid j}\frac{(x+t)^j(t-x)^j}{8^j(j!)^2}m^{2j}\right) = \\
&=\frac{m}{2}J_0\left(m \sqrt{\frac{t^2-x^2}{2}}\right).
\end{align*}
\normalsize
Here the first equality holds by Remark~\ref{Ueq}, the second one follows from Proposition~\ref{formula} applied to $\xi=2A-1$ and $\eta = 2B$, and the equivalence follows from Lemma~\ref{norm}. The next equality holds because all three sums are finite, while the limit transition follows from Lemmas~\ref{lim20}--\ref{binom} (analogues of \cite[Lemma 20 and Lemma 21]{article}) and \cite[Lemma 22]{article}. Lemmas~\ref{norm}--\ref{binom} are proved below.

For $a_2(x,t,m,\varepsilon,u_{\varepsilon})$,  the theorem is proved analogously.
$\square$

\begin{lemma}
\label{norm}
For each $m,t > 0$ we have $$\lim_{\varepsilon \to 0}(1+(m\varepsilon)^2)^{3/2-2\left\lfloor\frac{t}{4\varepsilon}\right\rfloor} = 1.$$ 
\end{lemma}
\begin{proof} This follows from the following chain of formulae:
$$
1 \le \left(1+(m\varepsilon)^2\right)^{2\left\lfloor\frac{t}{4\varepsilon}\right\rfloor - 3/2} \le (1+(m\varepsilon)^2)^{\frac{ t}{2\varepsilon}} = \left(1+(m\varepsilon)^2\right)^{\frac{t\varepsilon}{2\varepsilon^2}} \sim \left(\exp{\frac{m^2t}{2}}\right)^{\varepsilon} \underset{\varepsilon \to 0}{\to} 1.$$
Applying the squeeze theorem we obtain the required result.
\end{proof}

\begin{lemma}
\label{lim20}
For each integer $r>0$ we have
$$\lim\limits_{\varepsilon \to 0}{A-1\choose r}{B-1 \choose r}(2 + (m\varepsilon)^2)^{r}m^{2r}\varepsilon^{2r} = \frac{(x+t)^{r}(x-t)^{r}m^{2r}}{2^{3r}(r!)^2}.$$
\end{lemma}
 
\begin{proof}
 We have 
\begin{align*}
 &{A-1 \choose r}{B-1 \choose r}(2 + (m\varepsilon)^2)^{r}m^{2r}\varepsilon^{2r} = \\ 
 &=\frac{(A-1)\dots(A-r)\cdot(B-1)\dots(B-r)}{(r!)^2}\cdot(2 + (m\varepsilon)^2)^{r}m^{2r}\varepsilon^{2r} \to \\
 &\to \left(\frac{x+t}{4}\right)^r\left(\frac{t-x}{4}\right)^r\frac{2^rm^{2r}}{(r!)^2}
 \end{align*}
as $\varepsilon \to 0$, because for each $1 \le j \le r$
\begin{align*}
    \lim\limits_{\varepsilon \to 0}(A-j)\varepsilon =  \lim\limits_{\varepsilon \to 0}\left(\left\lfloor\frac{x}{4\varepsilon}\right\rfloor + \left\lfloor\frac{t}{4\varepsilon}\right\rfloor - j\right)\varepsilon = \lim\limits_{\varepsilon \to 0}\left(\frac{x}{4\varepsilon} + \frac{t}{4\varepsilon}\right)\varepsilon = \frac{x+t}{4}.
\end{align*}
Analogously, $\lim\limits_{\varepsilon \to 0}(B-j)\varepsilon = \frac{t-x}{4}.$
\end{proof} 

\begin{lemma}
\label{binom}
For each $|x|<t$ for small enough $\varepsilon > 0$ for each integer $j \ge 0$ we have
$$
{A - 1\choose j}{B - 1 \choose j}(2 + (m\varepsilon)^2)^j(m\varepsilon)^{2j} \le \frac{(x+t)^j(t-x)^jm^{2j}}{2^{3j}(j!)^2}.
$$
\end{lemma}
\begin{proof} The lemma follows from the following chain of inequalities:
\begin{align*}
&{A - 1\choose j}{B - 1 \choose j}(2 + (m\varepsilon)^2)^j(m\varepsilon)^{2j} \le \\
&\le \frac{\left(x+t-4\varepsilon\right)^j\left(t-x-4\varepsilon\right)^j}{(4\varepsilon)^{2j}(j!)^2}\left(2 + (m\varepsilon)^2\right)^j(m\varepsilon)^{2j} =\\
&=\frac{\left(x+t-4\varepsilon\right)^j\left(t-x-4\varepsilon\right)^j}{(4)^{2j}(j!)^2}\left(2 + (m\varepsilon)^2\right)^jm^{2j} \le \frac{(x+t)^j(t-x)^jm^{2j}}{2^{3j}(j!)^2},
\end{align*}
where the last inequality holds for $\varepsilon \le \min\left\{\frac{8}{m^2(x+t)}, \frac{x+t}{4}, \frac{t-x}{4}\right\},$
because for such $\varepsilon$ the following chain of inequalities holds:
$$
\left(1 - \frac{4\varepsilon}{x+t}\right) \left(1 + \frac{(m\varepsilon)^2}{2}\right) \le 1-\frac{4\varepsilon}{x+t} + \frac{(m\varepsilon)^2}{2}\le 1.
$$

\end{proof}
\textbf{Proof of Theorem~\ref{fproblim}.}
 It follows directly from Theorem~\ref{wavelimit}.
$\square$

\section{Proof of Theorems~\ref{wl-intr}~and~\ref{large-time-lim}}
We use the following notation for the \bluer{''normalized``} integrands from Proposition~\ref{IntForm}:\remove{}\begin{align}
\begin{split}
 \label{Fourier_form}
    \hat a_1(p,t)& :=
    \begin{dcases}
    2m\varepsilon\sin{\omega_p(t-\varepsilon)}\cdot \frac{\sin{p\varepsilon}}{(1+m^2\varepsilon^2)\cos{2\omega_p\varepsilon}}, &\text{ if } \frac{t}{\varepsilon} \equiv_{4} 1 ;\\
    m\varepsilon\cdot\frac{i\sin{\omega_p(t-2\varepsilon) - \cos{\omega_pt} }}{\sqrt{1+m^2\varepsilon^2}\cos{2\omega_p\varepsilon}}, &\text{ if } \frac{t}{\varepsilon} \equiv_{4} 2;\\
     2im\varepsilon\cos{\omega_p(t-\varepsilon)}\cdot \frac{\sin{p\varepsilon}}{(1+m^2\varepsilon^2)\cos{2\omega_p\varepsilon}}, &\text{ if } \frac{t}{\varepsilon} \equiv_{4} 3;\\
     m\varepsilon\cdot\frac{\cos{\omega}_p(t-2\varepsilon) - i\sin{\omega}_pt}{\sqrt{1+m^2\varepsilon^2}\cos{2\omega_p\varepsilon}}, &\text{ if } \frac{t}{\varepsilon} \equiv_{4} 0;
    \end{dcases}\\
    \hat{a}_2(p,t) &:= 
    \begin{dcases}
     \bluevar{e^{-ip\varepsilon}}\left(i\sin{\omega_p(t-\varepsilon)}\cdot\frac{m^2\varepsilon^2 + \cos{2p\varepsilon}}{(1+m^2\varepsilon^2)\cos{2\omega_p\varepsilon}}  - \cos{\omega_p(t-\varepsilon)}\right), &\text{ if } \frac{t}{\varepsilon} \equiv_{4} 1;\\
    \bluevar{e^{-2ip\varepsilon}}\cdot\frac{ie^{2ip\varepsilon}\sin{\omega_p(t-2\varepsilon)} - \cos{\omega_pt} }{\sqrt{1+m^2\varepsilon^2}\cos{2\omega_p\varepsilon}}, &\text{ if } \frac{t}{\varepsilon} \equiv_{4} 2;\\
     \bluevar{e^{-ip\varepsilon}}\left(\cos{\omega_p(t-\varepsilon)}\cdot\frac{m^2\varepsilon^2 + \cos{2p\varepsilon}}{(1+m^2\varepsilon^2)\cos{2\omega_p\varepsilon}}  - i\sin{\omega_p(t-\varepsilon)}\right), &\text{ if } \frac{t}{\varepsilon} \equiv_{4} 3;\\
     \bluevar{e^{-2ip\varepsilon}}\cdot
    \frac{e^{2ip\varepsilon}\cos{\omega_p(t-2\varepsilon)} - i\sin{\omega_pt}}{\sqrt{1+m^2\varepsilon^2}\cos{2\omega_p\varepsilon}}, &\text{ if } \frac{t}{\varepsilon} \equiv_{4} 0;
    \end{dcases}
\end{split}
\end{align}
\normalsize
Then $a_1(x,t,m,\varepsilon,u_{\varepsilon})$ and $a_2(x,t,m,\varepsilon,u_{\varepsilon})$ are up to sign the coefficients of the Fourier series for functions $\hat a_1(p,t)$ and $\hat a_2(p,t)$ respectively.

\textbf{Proof of Theorem~\ref{large-time-lim}.}
By~\cite[Lemma 2]{article} it suffices to prove (C). It follows from the following chain of equalities, which we are going to comment below:

\begin{align*}
    &\sum\limits_{x \in \varepsilon\mathbb{Z}}\frac{x^r}{t^r}P(x,t,m,\varepsilon,u_{\varepsilon})= \sum\limits_{x \in\varepsilon\mathbb{Z}}\begin{pmatrix}
    a_1(x,t,m,\varepsilon,u_{\varepsilon})\\
    a_2(x,t,m,\varepsilon,u_{\varepsilon})
    \end{pmatrix}^{*}\frac{x^r}{t^r}\begin{pmatrix}
    a_1(x,t,m,\varepsilon,u_{\varepsilon})\\
    a_2(x,t,m,\varepsilon,u_{\varepsilon}) \\
    \end{pmatrix} = \\
    &=
    \frac{\varepsilon}{2\pi}\int\limits_{-\pi/\varepsilon}^{\pi/\varepsilon}\begin{pmatrix}
    \hat{a}_1(p,t)\\
    \hat{a}_2(p,t)
    \end{pmatrix}^{*}\frac{i^r}{t^r}\frac{\partial^r}{\partial p^r}\begin{pmatrix}
    \hat{a}_1(p,t)\\
    \hat{a}_2(p,t) \\
    \end{pmatrix}d p =\\
     &=\begin{dcases}
        \frac{\varepsilon}{2\pi}\int\limits_{-\pi/\varepsilon}^{\pi/\varepsilon}\left( \omega_p^{\prime}\right)^{r} 
\left|\begin{pmatrix}
    \hat{a}_1(p,t)\\
    \hat{a}_2(p,t)
    \end{pmatrix}\right|^{2} dp + O_{m,\varepsilon,r}\left(\frac{1}{t}\right), & \text{ for even } r,\\
    \frac{\varepsilon}{2\pi}\int\limits_{-\pi/\varepsilon}^{\pi/\varepsilon}\left( \omega_p^{\prime}\right)^{r-1}
\begin{pmatrix}
    \hat{a}_1(p,t)\\
    \hat{a}_2(p,t)
    \end{pmatrix}^{*}\bluevar{\frac{i}{t}}\frac{\partial}{\partial p}\begin{pmatrix}
    \hat{a}_1(p,t)\\
    \hat{a}_2(p,t)
    \end{pmatrix} dp + O_{m,\varepsilon,r}\left(\frac{1}{t}\right),& \text{ for odd } r,
    \end{dcases} =
\end{align*}
\begin{align*}
    &=\begin{dcases}
       \frac{\varepsilon}{2\pi}\int\limits_{-\pi/\varepsilon}^{\pi/\varepsilon}\left( \omega_p^{\prime} \right)^{r}dp +O_{m,\varepsilon,r}\left(\frac{1}{t}\right),& \text{ for even } r,\\
       \frac{\varepsilon}{2\pi}\int\limits_{-\pi/\varepsilon}^{\pi/\varepsilon}\left( \omega_p^{\prime} \right)^{r} \frac{m^2\varepsilon^2 + \cos 2p\varepsilon}{(1+m^2\varepsilon^2)\cos 2\omega_p\varepsilon} dp+O_{m,\varepsilon,r}\left(\frac{1}{t}\right),& \text{ for odd } r,
    \end{dcases} =\\
    &=\frac{\varepsilon}{2\pi}\int\limits_{-\pi/\varepsilon}^{\pi/\varepsilon}\left( \omega_p^{\prime} \right)^{r+\delta_2(r)}dp +O_{m,\varepsilon,r}\left(\frac{1}{t}\right) = \frac{2\varepsilon}{\pi}\int\limits_{0}^{\pi/2\varepsilon}\left( \omega_p^{\prime} \right)^{r+\delta_2(r)}dp +O_{m,\varepsilon,r}\left(\frac{1}{t}\right)=\\
&=\int\limits_{-\frac{1}{1+m^2\varepsilon^2}}^{\frac{1}{1+m^2\varepsilon^2}}
    v^{r+\delta_2(r)}\frac{\sqrt{(1+m^2\varepsilon^2)^2-1}}{\pi(1-v^2)\sqrt{1 - (1+m^2\varepsilon^2)^2v^2}} dv + O_{m,\varepsilon,r}\left(\frac{1}{t}\right) =\\ &=\int\limits_{-\frac{1}{1+m^2\varepsilon^2}}^{\frac{1}{1+m^2\varepsilon^2}}
    (v^{r+\delta_2(r)}+v^{r+\delta_2(r+1)})\frac{\sqrt{(1+m^2\varepsilon^2)^2-1}}{\pi(1-v^2)\sqrt{1 - (1+m^2\varepsilon^2)^2v^2}}dv +O_{m,\varepsilon,r}\left(\frac{1}{t}\right) = \\
    &=\int\limits_{-1}^{1}
    v^{r}F^{\prime}(v) dv +O_{m,\varepsilon,r}\left(\frac{1}{t}\right).
\end{align*}
Here the first equality is true by the definition of $P(x,t,m,\varepsilon,u_{\varepsilon})$. The second one follows from Proposition~\ref{IntForm}, derivative property of Fourier series and Parseval's theorem. The third one follows from the following asymptotic formula for the derivative:
\begin{align*}
    \frac{1}{t^r}\cdot\frac{\partial^r}{\partial p^r} \hat{a}_k(p,t) 
   = \begin{cases}
   (-1)^{r/2}\left( \omega_p^{\prime} \right)^r \hat{a}_k(p,t) +O_{m,\varepsilon,r}\left(\frac{1}{t}\right),& \text{for } r \text{ even},\\
   (-1)^{(r-1)/2}\left( \omega_p^{\prime} \right)^{r-1} \bluevar{\frac{1}{t}}\frac{\partial}{\partial p}\hat{a}_k(p,t) +O_{m,\varepsilon,r}\left(\frac{1}{t}\right), & \text{for } r \text{ odd}.
   \end{cases}
\end{align*}
Indeed, we differentiate expression~(\ref{Fourier_form}) $r$ times with \bluer{the} help of Leibniz rule. In order to obtain the main term, \bluer{each time} we have to differentiate the factor containing $t$, that is, either $\sin\omega_p(t-k\varepsilon)$ or $\cos\omega_p(t-k\varepsilon)$, where $k=0,1$, or $2$. The remaining terms give $O_{m,\varepsilon,r}\left(\frac{1}{t}\right)$, since they are smooth periodic functions.  

The fourth equality is verified by a direct computation in \cite[Section 2]{computations}. The fifth and the sixth equalities follow from the equalities $\omega_p^{\prime} =\frac{\cos 2p\varepsilon}{(1+m^2\varepsilon^2)\cos2\omega_p\varepsilon} $ and $\omega_{p+\pi /2\varepsilon}^{\prime} = -\omega_p^{\prime} $.
The seventh one is obtained by the change of variable $v = \omega_p^{\prime}=\frac{\cos 2p\varepsilon}{\sqrt{(1+m^2\varepsilon^2)^2 - \sin^2 2p\varepsilon}}$ (see \cite[Sections 1 and 4]{computations}).
Such change of variable is possible, because $\omega_p^{\prime}$ decreases on $(0, \pi/2\varepsilon)$ (see \cite[Section 1]{computations}). The eighth equality holds because the following function is odd and does not affect the integral:
$$
v^{r+\delta_2(r+1)}\frac{\sqrt{(1+m^2\varepsilon^2)^2-1}}{\pi(1-v^2)\sqrt{1 - (1+m^2\varepsilon^2)^2v^2}}.
$$
Thus Theorem~\ref{large-time-lim} is proved. $\square$
\textbf{Proof of Theorem~\ref{wl-intr}.} It follows directly from Theorem~\ref{large-time-lim}. $\square$

\section{Proof of Theorems~\ref{chir-rev-f} and \ref{prob-rev-field}}
\textbf{Proof of Theorem \ref{prob-rev-field} modulo Lemma~\ref{cos-approx1} below.} Consider 4 cases:

\textit{Case 1:} $\frac{t}{\varepsilon} \equiv_{4} 1$.
We have the following chain of equalities:
\begin{align*}
    &\sum\limits_{x \in \varepsilon\mathbb{Z}}a_1^2(x,t,m,\varepsilon,u_{\varepsilon})= \sum\limits_{x \in\varepsilon\mathbb{Z}}\left|a_1(x,t,m,\varepsilon, u_{\varepsilon})\right|^2 = \\
    &=
    \frac{\varepsilon}{2\pi}\int\limits_{-\pi/\varepsilon}^{\pi/\varepsilon} \left|\hat{a}_1(p,t)\right|^2 d p =   \frac{\varepsilon}{2\pi}\int\limits_{-\pi/\varepsilon}^{\pi/\varepsilon} \frac{4m^2\varepsilon^2 \sin^2 \omega_p(t-\varepsilon) \sin^2 p\varepsilon}{(1+m^2\varepsilon^2)^2 \cos^2 2\omega_p\varepsilon} d p= \\
      &=
    \frac{\varepsilon}{2\pi}\int\limits_{-\pi/\varepsilon}^{\pi/\varepsilon} \frac{2m^2\varepsilon^2 \sin^2 \omega_p(t-\varepsilon)}{(1+m^2\varepsilon^2)^2 \cos^2 2\omega_p\varepsilon} d p =  \frac{\varepsilon}{2\pi}\int\limits_{-\pi/\varepsilon}^{\pi/\varepsilon} \frac{m^2\varepsilon^2  d p}{(1+m^2\varepsilon^2)^2 \cos^2 2\omega_p\varepsilon} -\\
    &-\frac{\varepsilon}{2\pi}\int\limits_{-\pi/\varepsilon}^{\pi/\varepsilon} \frac{m^2\varepsilon^2\cos 2 \bluevar{\omega_p(t-\varepsilon)} d p}{(1+m^2\varepsilon^2)^2 \cos^2 2\omega_p\varepsilon} = \frac{\varepsilon}{2\pi}\int\limits_{-\pi/\varepsilon}^{\pi/\varepsilon} \frac{m^2\varepsilon^2d p}{(1+m^2\varepsilon^2)^2 \cos^2 2\omega_p\varepsilon}  + \\
    &+O_{m,\varepsilon}\left(t^{-1/3}\right)=\frac{2\varepsilon}{\pi}\int\limits_{-\pi/4\varepsilon}^{\pi/4\varepsilon} \frac{m^2\varepsilon^2d p}{(1+m^2\varepsilon^2)^2 - \sin^2 2p\varepsilon}  + O_{m,\varepsilon}\left(t^{-1/3}\right) 
    =\\
    &=\frac{m\varepsilon}{(1+m^2\varepsilon^2)\sqrt{2+m^2\varepsilon^2}} + O_{m,\varepsilon}\left(t^{-1/3}\right).
\end{align*}
The first equality holds because $a_1(x,t,m,\varepsilon, u_{\varepsilon}) \in \mathbb{R}$. The second equality follows from Proposition~\ref{IntForm}, notation \eqref{Fourier_form}, and Parseval's theorem. The third one follows from the definition of $\hat{a}_1(p,t)$. The fourth one follows from the power-reduction formula and the computation $\omega_{p +\pi/2\varepsilon} = \frac{1}{2\varepsilon}\arcsin \frac{\sin{(2p\varepsilon + \pi)}}{1+m^2\varepsilon^2} = -\omega_p$. The fifth one follows from the power-reduction formula and the linearity of the integral. The sixth one follows from Lemma \ref{cos-approx1}, which we prove below. The seventh one follows from $\omega_p=-\omega_{p+\pi/2\varepsilon}$ and the definition of $\omega_p$. The eighth one is a direct computation performed in \cite[Section 3]{computations}.

\textit{Case 2:} $\frac{t}{\varepsilon} \equiv_{4} 2$.
We have the following chain of equalities:
\begin{align*}
    &\sum\limits_{x \in \varepsilon\mathbb{Z}}a_1^2(x,t,m,\varepsilon,u_{\varepsilon})= \sum\limits_{x \in\varepsilon\mathbb{Z}}\left|a_1(x,t,m,\varepsilon, u_{\varepsilon})\right|^2 =\\
    &=
    \frac{\varepsilon}{2\pi}\int\limits_{-\pi/\varepsilon}^{\pi/\varepsilon} \left|\hat{a}_1(p,t) \right|^2d p =   \frac{\varepsilon}{2\pi}\int\limits_{-\pi/\varepsilon}^{\pi/\varepsilon} \frac{m^2\varepsilon^2 (\sin^2 \omega_p(t-2\varepsilon) +\cos^2 \omega_pt)}{(1+m^2\varepsilon^2) \cos^2 2\omega_p\varepsilon} d p=
\end{align*}
\begin{align*}
      &=
    \frac{2\varepsilon}{\pi}\left(\int\limits_{-\pi/4\varepsilon}^{\pi/4\varepsilon} \frac{m^2\varepsilon^2 d p }{(1+m^2\varepsilon^2) \cos^2 2\omega_p\varepsilon}  - \int\limits_{-\pi/4\varepsilon}^{\pi/4\varepsilon} \frac{m^2\varepsilon^2 \cos2 \omega_p(t-2\varepsilon) d p}{2(1+m^2\varepsilon^2) \cos^2 2\omega_p\varepsilon}  + \right. \\
    &+ \left. \int\limits_{-\pi/4\varepsilon}^{\pi/4\varepsilon} \frac{m^2\varepsilon^2 \cos 2\omega_pt \,d p}{2(1+m^2\varepsilon^2) \cos^2 2\omega_p\varepsilon}  \right) =\\
    &=\frac{2\varepsilon}{\pi}\int\limits_{-\pi/4\varepsilon}^{\pi/4\varepsilon} \frac{(1+m^2\varepsilon^2) m^2\varepsilon^2 d p}{(1+m^2\varepsilon^2)^2 - \sin^2 2p\varepsilon}+O_{m,\varepsilon}\left(t^{-1/3}\right)=\frac{m\varepsilon}{\sqrt{2+m^2\varepsilon^2}} + O_{m,\varepsilon}\left(t^{-1/3}\right).
\end{align*}
\normalsize
The first two equalities are analogous to Case 1. The third one follows from notation \eqref{Fourier_form}. The fourth one follows from the power-reduction formula and the computation $\omega_{p +\pi/2\varepsilon} = \frac{1}{2\varepsilon}\arcsin \frac{\sin{(2p\varepsilon + \pi)}}{1+m^2\varepsilon^2} = -\omega_p$. The fifth one follows from Lemma \ref{cos-approx1}, which we prove later, and the definition of $\omega_p$. The sixth equality is analogous to the seventh one in Case 1.

\textit{Cases $\frac{t}{\varepsilon}\equiv_{4} 3$ and $\frac{t}{\varepsilon}\equiv_{4} 0$} are analogous to Cases 1 and 2 respectively, only $\sin \omega_p(t-k\varepsilon)$ is replaced by $\cos \omega_p(t-k\varepsilon)$, where $k = 1$ and $2$, and vice versa. 
To finish the proof it remains to prove Lemma~\ref{cos-approx1}, which is a rough estimate in the spirit of \bluer{the} stationary phase method \cite[Lemma 3]{article}. $\square$
\begin{lemma}
\label{cos-approx1} 
Assume $m,\varepsilon,t >0$, then 
\begin{align*}
    \int\limits_{-\pi/4\varepsilon}^{\pi/4\varepsilon} \frac{\cos2 \omega_pt}{\cos^2 2\omega_p\varepsilon} d p = O_{m,\varepsilon}\left(t^{-1/3}\right).
\end{align*}
\end{lemma}
\begin{proof}
Take $t > (\frac{4\varepsilon}{\pi})^3$, otherwise Lemma is obvious, because the integrand is bounded from above:
\begin{align*}
    \left|\frac{\cos2 \omega_pt}{\cos^2 2\omega_p\varepsilon}\right| \le \left|\frac{1}{1 - \frac{\sin^2 2p\varepsilon}{(1+m^2\varepsilon^2)^2}}\right| \le \frac{(1+m^2\varepsilon^2)^2}{(1+m^2\varepsilon^2)^2 - 1}.
\end{align*}
\normalsize
Let us split the integral into three parts $I_1, I_2, I_3$:
\begin{align*}
    I_1 := \hspace{-0.5cm}\int\limits_{-\pi/4\varepsilon}^{-\pi/4\varepsilon + t^{-1/3}} \hspace{-0.5cm}\frac{\cos2 \omega_pt}{\cos^2 2\omega_p\varepsilon} d p,~ 
    I_2 := \hspace{-0.5cm}\int\limits_{-\pi/4\varepsilon+t^{-1/3}}^{\pi/4\varepsilon-t^{-1/3}} \hspace{-0.5cm}\frac{\cos2 \omega_pt}{\cos^2 2\omega_p\varepsilon} d p,~
    I_3 := \hspace{-0.5cm}\int\limits_{\pi/4\varepsilon-t^{-1/3}}^{\pi/4\varepsilon} \hspace{-0.5cm}\frac{\cos2 \omega_pt}{\cos^2 2\omega_p\varepsilon} d p.
\end{align*}
Note that $I_1 + I_3 = O_{m,\varepsilon}(t^{-1/3})$, because integrand is bounded from above by $\frac{(1+m^2\varepsilon^2)^2}{(1+m^2\varepsilon^2)^2 - 1}$.

Let us find an asymptotic formula for $I_2$. Multiply and divide the integrand by $\omega_p^{\prime}$ and integrate by parts. This operation is possible because $\omega_p^{\prime} = \frac{\cos 2p\varepsilon}{(1+m^2\varepsilon^2)\cos2\omega_p\varepsilon} $ does not vanish on $[-\pi/4\varepsilon + t^{-1/3}, \pi/4\varepsilon - t^{-1/3}]$.

We get
\begin{align}
\label{integr}
    &I_2 = \left.\frac{\sin2\omega_pt}{2t\omega_p^{\prime}\cos^2 2\omega_p\varepsilon }\right|^{\pi/4\varepsilon-t^{-1/3}}_{-\pi/4\varepsilon+t^{-1/3}} - \int\limits^{\pi/4\varepsilon-t^{-1/3}}_{-\pi/4\varepsilon+t^{-1/3}} \frac{\sin2\omega_pt}{2t}\cdot\left(\frac{1}{\omega_p^{\prime}\cos^2 2\omega_p\varepsilon}\right)^{\prime} dp
\end{align}
The first term in the right side of \eqref{integr} is estimated as follows:
\begin{align*}
      & \left.\frac{\sin2\omega_pt}{2t\omega_p^{\prime}\cos^2 2\omega_p\varepsilon }\right|^{\pi/4\varepsilon-t^{-1/3}}_{-\pi/4\varepsilon+t^{-1/3}} = \left.\frac{(1+m^2\varepsilon^2)^2 \sin2\omega_{p}t}{t\cos2p\varepsilon\sqrt{(1+m^2\varepsilon^2)^2 - \sin^2 2p\varepsilon}}\right|_{p=-\pi/4\varepsilon+t^{-1/3}} \le\\
      & \le \frac{(1+m^2\varepsilon^2)^2}{2t\sin(2\varepsilon t^{-1/3})\sqrt{(1+m^2\varepsilon^2)^2 - 1}} \le \frac{(1+m^2\varepsilon^2)^2}{2\varepsilon t^{2/3}\sqrt{(1+m^2\varepsilon^2)^2 - 1}} = O_{m,\varepsilon}\left(t^{-2/3}\right).
\end{align*}
Here the first equality follows from the definition of $\omega_p^{\prime}$. The second one holds because the sine is bounded. The third one follows from $|\sin{2\varepsilon t^{-1/3}}| \ge |\varepsilon t^{-1/3}|$ for $t > \left(\frac{4\varepsilon}{\pi}\right)^3$. Note, that we made an assumption that the first term in the right side of \eqref{integr}
is non-negative. The case of non-positive term is analogous.

To estimate the second term in the right side of the \eqref{integr},
note that
\begin{align*}
    &\left(\frac{1}{\omega_p^{\prime}\cos^2 2\omega_p\varepsilon}\right)^{\prime} = \left(\frac{1+m^2\varepsilon^2}{\cos 2p\varepsilon\cos 2\omega_p\varepsilon}\right)^{\prime}= \\
    &=2\varepsilon(1+m^2\varepsilon^2)\cdot\frac{\sin2p\varepsilon}{\cos^2 2p\varepsilon\cos 2\omega_p\varepsilon} + 2\varepsilon \frac{\sin2\omega_p\varepsilon}{\cos^3 2\omega_p\varepsilon}.
\end{align*}
Since $|\cos 2\omega_p\varepsilon| \ge \frac{(1+m^2\varepsilon^2)^2}{(1+m^2\varepsilon^2)^2-1}$ and $$\cos^2{2p\varepsilon} \ge \cos^2{(\pi/2 - 2\varepsilon t^{-1/3})} = \sin^2(2\varepsilon t^{-1/3}) \ge \varepsilon^2t^{-2/3}$$ on $[-\pi/4\varepsilon + t^{-1/3}, \pi/4\varepsilon - t^{-1/3}]$ for $t \ge \left(\frac{4\varepsilon}{\pi}\right)^3$, it follows that $I_2 = O_{m,\varepsilon}\left(t^{-1/3}\right)$.

Now, adding the bounds for $I_1, I_2,I_3$, we obtain the required bound.
\end{proof}

\textbf{Proof of Theorem~\ref{chir-rev-f}}
It follows directly from Theorem~\ref{prob-rev-field}. $\square$

\section{Proof of Theorem \ref{Airy}}

In order to prove Theorem \ref{Airy} we need the following result, obtained by P. Zakorko.

\begin{theorem}[{\cite[Theorem 3]{Zakorko}}]
\label{Airy_Zak}
Let $f \colon \left[ -U,U\right] \times \left[0, A\right] \to \mathbb{R}$ be a real analytic function such that
\begin{itemize}
    \item[1.] $f'_u(0,0)=0, f''_{uu}(0,0)=0, f'''_{uuu}\ne 0, f''_{\alpha u}(0,0) \ne 0$;
    \item[2.] $f(u,\alpha) = -f(-u,\alpha)$ for each $u \in \left[ -U,U\right]$ and $\alpha \in \left[0,A\right]$;
    \item[3.] for each $\alpha \in \left(0, A\right]$ there are precisely two solutions $\pm u_0(\alpha)\in \left[-U, U\right]$ of the equation $f'_u(u,\alpha)=0$, where $u_0(\alpha) \in \left( 0, U\right)$, and we have $f''_{uu}(u_0(\alpha),\alpha)>0$;
    \item[4.] for $\alpha = 0$ there is only one solution $u_0(0)=0$ of the equation $f'_u(u,\alpha)=0$.
\end{itemize}
Let $g \in C^{\infty}\left[-U,U\right]$ be even, and $t$ be a positive real number. Then for each $\alpha \ge 0$
\begin{align*}
    \int\limits_{-U}^U g(u)e^{it f(u,\alpha)}du = \frac{2\pi g(u_0\bluevar{(\alpha)})}{t^{1/3}h(\alpha)}\mathrm{Ai}\left(-\left(-\frac{3}{2}f(u_0(\alpha),\alpha) t\right)^{2/3}\right) + O_{f,g}\left(\frac{1}{t}\right),
\end{align*}
where 
$$h(\alpha) := 
\begin{dcases}
    \frac{f''_{uu}(u_0(\alpha),\alpha)^{1/2}}{(-12f(u_0(\alpha),\alpha))^{1/6}}, &\text{ if } 0 < \alpha \le A,\\
     \left(\frac{1}{2}f'''_{uuu}(0,0)\right)^{1/3}, &\text{if } \alpha = 0.
\end{dcases}$$
\end{theorem}
\bluer{Here $O_{f,g}\left(\frac{1}{t}\right)$ for two functions $f,g$ is defined analogously to $O_{m,\varepsilon}\left(\frac{1}{t}\right)$.}
We have the following lemma.

\begin{lemma}
\label{con_check}
    For $u \in \left[-\pi/2\varepsilon, \pi/2\varepsilon \right]$ and $\alpha \in \left[0, \frac{1}{1+m^2\varepsilon^2}\right]$ the function $$f(u,\alpha) := u\left(\frac{1}{1+m^2\varepsilon^2} - \alpha\right) - \frac{1}{2\varepsilon}\arcsin \frac{\sin 2u\varepsilon}{1+m^2\varepsilon^2}$$ satisfies conditions 1--4 of Theorem \ref{Airy_Zak}.
\end{lemma}

\begin{proof}
    The direct computations in \cite[Section 5]{computations} show that
    \begin{align*}
        f^{\prime}_u(u,\alpha) &= \frac{1}{1+m^2\varepsilon^2}-\alpha - \frac{1}{1+m^2\varepsilon^2}\cdot\frac{\cos 2u\varepsilon}{\cos 2\omega_u \varepsilon};\\
        f^{\prime\prime}_{uu}(u,\alpha) &= \frac{(1+m^2\varepsilon^2)^2 -1 }{(1+m^2\varepsilon^2)^3}\cdot\frac{2\varepsilon \sin 2u\varepsilon}{\cos^3 2\omega_u\varepsilon}; \\
         f^{\prime\prime\prime}_{uuu}(u,\alpha) &= \frac{(1+m^2\varepsilon^2)^2 -1 }{(1+m^2\varepsilon^2)^5}\cdot\frac{4\varepsilon^2 \cos 2u\varepsilon \left((1+m^2\varepsilon^2)^2 + 2\sin^2 2u\varepsilon\right)}{\cos^{5} 2\omega_u\varepsilon}; \\
         f^{\prime\prime}_{u\alpha} &= -1;\\
         \pm u_0(\alpha) &= \pm\frac{1}{2\varepsilon}\arccos\left( v \sqrt{\frac{(1+m^2\varepsilon^2)^2-1}{1-v^2}}\right).
    \end{align*}
    Here $v:= \frac{1}{1+m^2\varepsilon^2} - \alpha$. So, condition 1 is satisfied, since $\omega_0 = 0$. Condition 2 is satisfied since $\omega_u = \frac{1}{2\varepsilon}\arcsin \frac{\sin 2\bluevar{u}\varepsilon}{1+m^2\varepsilon^2} = - \frac{1}{2\varepsilon}\arcsin \frac{\sin (-2\bluevar{u}\varepsilon)}{1+m^2\varepsilon^2} = -\omega_{-u}$. Condition 3 holds since arccosine is a bijection between $[0,1]$ and $[0, \pi/2]$ and $f^{\prime\prime}_{uu}(u,\alpha) > 0 $ for $0<u,\alpha\bluevar{<\frac{\pi}{2\varepsilon}}$. Condition 4 is obvious.
\end{proof}

\textbf{Proof of Theorem \ref{Airy}}
We have the following chain of equalities:
\begin{align*}
    &a_1(x,t,m,\varepsilon, u_{\varepsilon}) = \\
    &=(-1)^{\frac{x}{2\varepsilon}+\left\lfloor \frac{x + t}{4\varepsilon}\right\rfloor}\frac{m\varepsilon^2}{2\pi\sqrt{1+m^2\varepsilon^2}}\int\limits^{\pi/\varepsilon}_{-\pi/\varepsilon} e^{iux} \frac{i\sin \omega_u (t-2\varepsilon) - \cos\omega_u t}{\cos 2\omega_u \varepsilon} du = \\
   &= (-1)^{\left\lfloor \frac{x + t+4\varepsilon}{4\varepsilon}\right\rfloor}\frac{m\varepsilon^2}{2\pi\sqrt{1+m^2\varepsilon^2}} \int\limits^{\pi/\varepsilon}_{-\pi/\varepsilon} \frac{e^{iux - \omega_u t}}{\cos 2\omega_u \varepsilon}du =\\
   &= (-1)^{\left\lfloor \frac{x + t+4\varepsilon}{4\varepsilon}\right\rfloor}\frac{m\varepsilon^2}{\pi\sqrt{1+m^2\varepsilon^2}} \int\limits^{\pi/2\varepsilon}_{-\pi/2\varepsilon} \frac{e^{iu(x/t - \omega_u) t}}{\cos 2\omega_u \varepsilon}du =\\
    &=(-1)^{\left\lfloor\frac{x+t+4\varepsilon}{4\varepsilon} \right\rfloor}\frac{\sqrt{2}m^{1/2}\varepsilon(1+m^2\varepsilon^2)^{1/2}\left(-12\tilde{\theta}(x/t)\right)^{1/6}}{((2 + m^2\varepsilon^2)(1-(x/t)^2(1+m^2\varepsilon^2)^2))^{1/4}}\cdot \\
    &\cdot\left(\frac{1}{t}\right)^{1/3} \mathrm{Ai }\left(-\left(-\frac{3}{2}\tilde{\theta} (x/t)t\right) ^{2/3}\right)+ O_{m,\varepsilon}\left(\frac{1}{t}\right).
\end{align*}
Here the first equality follows from Proposition \ref{IntForm}. \bluer{The} second and \bluevar{the} third equalities hold since $\omega_{\bluevar{u}} = -\omega_{\bluevar{u}+\pi/2\varepsilon}$ and $e^{i(\bluevar{u}+\pi/2\varepsilon)x}= e^{i\bluevar{u}x}$ for $x \in 4\varepsilon\mathbb{Z}$. The last one is application of Theorem \ref{Airy_Zak} and Lemma \ref{con_check} for $\alpha = \frac{1}{1+m^2\varepsilon^2} - \frac{x}{t}$ and $g(u) = \frac{1}{\cos 2\omega_{\bluevar{u}} \varepsilon}$. $\square$

\section{Acknowledgements}
The work was supported by Russian Science Foundation, \newline grant No 22-41-05001, https://rscf.ru/en/project/22-41-05001/. \\
\bluer{The author is grateful to A. Ustinov, F. Kuyanov, and M. Dmitriev for useful discussions. Especially author wants to thank his scientific supervisor M. Skopenkov for numerous remarks and invaluable help in improving the quality of this text.}

\appendix
\section{Second proof of Proposition~\ref{formula}}
\label{first-proof}
\textbf{Proof of Proposition~\ref{formula}.}
We prove it by induction on~$\xi+\eta$. \bluer{Let us denote $n:=1 + m^2\varepsilon^2$ in order to simplify formulae below.}\\

\textit{Base} of induction $\xi=\eta=0$ is obvious. 
%На рис.1 можно видеть вычисленные значения $a_1(\xi - \eta + 1,\xi + \eta + 1,m,u)$ и $a_2(\xi - \eta + 1, \xi + \eta + 1,m , u)$ при малых $k,n$. (здесь вставлю табличку со значениями для малых $n,k$)

\textit{Step.} We have: \par
 \begin{align*}
  &a_1(\xi - \eta + 1, \xi + \eta + 1, m,1 ,u_1) = \\
  &=\frac{a_1(\xi - \eta + 2, \xi + \eta, m,1, u_1) + ma_2(\xi - \eta + 2, \xi + \eta, m,1, u_1)}{\sqrt{n}}=\\
  &=\frac{(-1)^{\xi+1}}{n^{\frac{\xi + \eta}{2}}}\left(mn^{\delta_2(\xi\eta) }\sum_{j=0}^{\lfloor\frac{\xi}{2}\rfloor}{\lfloor\frac{\xi}{2}\rfloor\choose j}{\lfloor\frac{\eta-2}{2}\rfloor\choose j}(1 - n^2)^j + \right.\\
 &+\left.m\sum_{j=0}^{\lfloor\frac{\xi}{2}\rfloor}\left( {\lfloor\frac{\eta-1}{2}\rfloor\choose j}n^{\delta_2(\xi(\eta-1))} - {\lfloor\frac{\eta-2}{2}\rfloor\choose j}n^{\delta_2(\xi\eta) }\right){\lfloor\frac{\xi}{2}\rfloor\choose j}(1 - n^2)^j \right)=\\
  &=\frac{(-1)^{\xi+1}m}{n^{\frac{\xi+\eta}{2}}}n^{\delta_2((\eta+1)\xi)}\sum_{j=0}^{\lfloor\frac{\xi}{2}\rfloor}{\lfloor\frac{\xi}{2}\rfloor\choose j}{\lfloor\frac{\eta-1}{2}\rfloor\choose j}(1 - n^2)^j.
  \end{align*}
\normalsize
Here the first equality holds by Dirac's equation (Lemma~\ref{dirak}) and the definition of the field $u_{
\varepsilon}$, the second one holds by the inductive hypothesis, the third one is just an expansion in $(1 - n^2)^j$.

Let us perform the induction step for $a_2(\xi - \eta+1, \xi + \eta+1,m,1,u_1)$. Consider 2 cases:

\textit{Case 1:} $\eta \equiv_2 1$ or $\xi \equiv_2 1$. Then

\begin{align*}
&a_2(\xi - \eta + 1, \xi + \eta + 1, m,1 ,u_1) =\\
&=\frac{(-1)^{\eta+1}(-ma_1(\xi - \eta, \xi + \eta, m,1, u_1) + a_2(\xi - \eta, \xi + \eta, m,1, u_1))}{\sqrt{n}}=\\
&=\frac{(-1)^{\xi+\eta+1}}{n^{\frac{\xi+\eta}{2}}}\left(-m^2\sum_{j=0}^{\lfloor \frac{\xi-1}{2}\rfloor}{\lfloor\frac{\xi-1}{2}\rfloor\choose j}{\lfloor\frac{\eta-1}{2}\rfloor\choose j}(1 - n^2)^j  + \right.\\
&+\left.\sum_{j=0}^{\lfloor\frac{\xi-1}{2}\rfloor}\left( {\lfloor\frac{\eta}{2}\rfloor\choose j}n^{\delta_2((\xi-1)\eta)} - {\lfloor\frac{\eta-1}{2}\rfloor\choose j}\right){\lfloor\frac{\xi-1}{2}\rfloor\choose j}(1 - n^2)^j \right)=\\
&=\frac{(-1)^{\xi+\eta+1}}{n^{\frac{\xi+\eta}{2}}}\sum_{j=0}^{\lfloor\frac{\xi-1}{2}\rfloor}\left( {\lfloor\frac{\eta}{2}\rfloor\choose j}n^{\delta_2((\xi-1)\eta)} - n{\lfloor\frac{\eta-1}{2}\rfloor\choose j}\right){\lfloor\frac{\xi-1}{2}\rfloor\choose j}(1 - n^2)^j =\\
&=\frac{(-1)^{\xi+1}}{n^{\frac{\xi+\eta}{2}}}\sum_{j=0}^{\lfloor\frac{\xi}{2}\rfloor}\left( {\lfloor\frac{\eta}{2}\rfloor\choose j}n^{\delta_2(\xi\eta)} - {\lfloor\frac{\eta-1}{2}\rfloor\choose j}n^{\delta_2(\xi(\eta+1))}\right){\lfloor\frac{\xi}{2}\rfloor\choose j}(1 - n^2)^j.
\end{align*}
\normalsize
Here the first equality, again, follows from Dirac's equation (Lemma~\ref{dirak}) and the definition of the field $u_{\varepsilon}$, the second one holds by the inductive hypothesis, the third one is an expansion in $(1 - n^2)^j$. In order to check the fourth one, we need to consider 3 subcases.

\textit{Subcase 1:} $\xi \equiv_2 \eta \equiv_2 1$. Then the left hand side and right hand side are equal term-wise:
\begin{align*}
&\frac{-1}{n^{\frac{\xi+\eta}{2}}}\sum_{j=0}^{\frac{\xi-1}{2}}\left( {\frac{\eta-1}{2}\choose j} - n{\frac{\eta-1}{2}\choose j}\right){\frac{\xi-1}{2}\choose j}(1 - n^2)^j =\\
&=\frac{1}{n^{\frac{\xi+\eta}{2}}}\sum_{j=0}^{\frac{\xi-1}{2}}\left( {\frac{\eta-1}{2}\choose j}n - {\frac{\eta-1}{2}\choose j}\right){\frac{\xi-1}{2}\choose j}(1 - n^2)^j.
\end{align*}
\normalsize

\textit{Subcase 2:} $\xi - 1 \equiv_2 \eta \equiv_2 1$. Then
\begin{align*}
&\frac{1}{n^{\frac{\xi+\eta}{2}}}\sum_{j=0}^{\frac{\xi-2}{2}}\left( n{\frac{\eta-1}{2}\choose j} - n{\frac{\eta-1}{2}\choose j}\right){\frac{\xi-2}{2}\choose j}(1 - n^2)^j = \\
 &=0=\frac{-1}{n^{\frac{\xi+\eta}{2}}}\sum_{j=0}^{\frac{\xi}{2}}\left( {\frac{\eta-1}{2}\choose j} - {\frac{\eta-1}{2}\choose j}\right){\frac{\xi}{2}\choose j}(1 - n^2)^j.
\end{align*}
 
\normalsize
\textit{Subcase 3:} $\xi \equiv_2 \eta-1 \equiv_2 1$. Then
\begin{align*}
&\frac{1}{n^{\frac{\xi+\eta}{2}}}\sum_{j=0}^{\frac{\xi-1}{2}}\left( {\frac{\eta}{2}\choose j} - n{\frac{\eta-2}{2}\choose j}\right){\frac{\xi-1}{2}\choose j}(1 - n^2)^j =\\
&=\frac{1}{n^{\frac{\xi+\eta}{2}}}\sum_{j=0}^{\frac{\xi-1}{2}}\left( {\frac{\eta}{2}\choose j} - {\frac{\eta-2}{2}\choose j}n\right){\frac{\xi-1}{2}\choose j}(1 - n^2)^j.
\end{align*}
\normalsize
This completes the proof of the first case. 

\textit{Case 2:} $\eta \equiv_2 \xi \equiv_2 0$. Then
\begin{align*}
&a_2(\xi - \eta + 1, \xi + \eta + 1, m, 1, u_1) = \\
&=\frac{(-1)^{\eta+1}(-ma_1(\xi - \eta, \xi + \eta, m,1, u) + a_2(\xi - \eta, \xi + \eta, m,1, u))}{\sqrt{n}}=\\
 &=\frac{(-1)^{\xi+\eta+1}}{n^{\frac{\xi+\eta}{2}}}\left(-m^2n\sum_{j=0}^{ \frac{\xi-2}{2}}{\frac{\xi-2}{2}\choose j}{\frac{\eta-2}{2}\choose j}(1 - n^2)^j  + \right.\\
&+\left.\sum_{j=0}^{\frac{\xi-2}{2}}\left( {\frac{\eta}{2}\choose j} - n{\frac{\eta-2}{2}\choose j}\right){\frac{\xi-2}{2}\choose j}(1 - n^2)^j \right)=\\
&=\frac{-1}{n^{\frac{\xi+\eta}{2}}}\sum_{j=0}^{\frac{\xi-2}{2}}\left( {\frac{\eta}{2}\choose j} - n^2{\frac{\eta-2}{2}\choose j}\right){\frac{\xi-2}{2}\choose j}(1 - n^2)^j =
\end{align*}
\begin{align*}
&=\frac{-1}{n^{\frac{\xi+\eta}{2}}}\left(\sum_{j=0}^{\frac{\xi-2}{2}}\left( {\frac{\eta}{2}\choose j} - {\frac{\eta-2}{2}\choose j}\right){\frac{\xi-2}{2}\choose j}(1 - n^2)^j \right. +\\
&+\left. \sum_{j=1}^{\frac{\xi}{2}}{\frac{\eta-2}{2}\choose j-1}{\frac{\xi-2}{2}\choose j-1}(1 - n^2)^j\right) =\\
&=\frac{-1}{n^{\frac{\xi+\eta}{2}}} \sum_{j=1}^{\frac{\xi}{2}} {\frac{\eta-2}{2}\choose j-1}\left({\frac{\xi-2}{2}\choose j}+{\frac{\xi-2}{2}\choose j-1}\right)(1 - n^2)^j  =\\
&=\frac{-1}{n^{\frac{\xi+\eta}{2}}}\sum_{j=0}^{\frac{\xi}{2}} \left({\frac{\eta}{2}\choose j}-{\frac{\eta-2}{2}\choose j}\right){\frac{\xi}{2}\choose j}(1 - n^2)^j. 
\end{align*}
 \normalsize
Here the first equality follows from Dirac's equation (Lemma~\ref{dirak}) and the definition of the field $u_{\varepsilon}$, the second one holds by the inductive hypothesis, the third one is an expansion in $(1 - n^2)^j$, the fourth one is obtained by adding and subtracting $$\sum\limits_{j=0}^{\frac{\xi-2}{2}}{\frac{\eta-2}{2}\choose j}{\frac{\xi-2}{2}\choose j}(1 - n^2)^j$$ and regrouping terms, the fifth one is the use of Pascal's rule and an expansion in $(1 - n^2)^j$, the sixth one is the application of Pascal's rule two times.

Thus Case 2 and hence the proposition is proved. $\square$

\section{Proof of Proposition~\ref{IntForm}}
\label{IntFormProof}
This is a direct computation by induction.

\textit{Case 1:} $x/\varepsilon$ and $t/\varepsilon$ have opposite parity.
Note that 
$$\omega_{p+\pi/\varepsilon} = \frac{1}{2\varepsilon}\arcsin \frac{\sin 2(p\varepsilon+\pi)}{1+m^2\varepsilon^2}= \frac{1}{2\varepsilon}\arcsin \frac{\sin 2p\varepsilon}{1+m^2\varepsilon^2} = \omega_{p}. $$
Hence, the integrand for $a_k(x,t,m,\varepsilon,u_{\varepsilon})$, where $k=1,2$, changes sign after change of variables $p \to p+\pi/\varepsilon$ (it is multiplied by $(-1)^{x/\varepsilon + 1}$ for odd $t/\varepsilon$ and by $(-1)^{x/\varepsilon}$ for even $t/\varepsilon$). Thus the integral equals $0$, and this completes the proof of the first case.

\textit{Case 2:} $x/\varepsilon$ and $t/\varepsilon$ have the same parity.
We prove the proposition by induction on $t/\varepsilon$. 

\textit{Base.} Let $t=\varepsilon$, $x\in \varepsilon\mathbb{Z}$. Since $\sin{\omega_p0}=0$, then the integrand for $a_1(x,t,m,\varepsilon,u_{\varepsilon})$ equals $0$. Thus, the induction base for $a_1(x,t,m,\varepsilon,u_{\varepsilon})$ is verified.
For $a_2(x,t,m,\varepsilon,u_{\varepsilon})$ we have
\begin{align*}
    &-\frac{\varepsilon(-1)^{\frac{x-\varepsilon}{2\varepsilon}+\left\lfloor\frac{x+\varepsilon}{4\varepsilon}\right\rfloor}}{2\pi}\int\limits_{-\pi/\varepsilon}^{\pi/\varepsilon}e^{ip(x-\varepsilon)}dp = \\
    &=\begin{dcases}
    -1, &\text{if  } x=\varepsilon,\\
    -\frac{\varepsilon(-1)^{\frac{x-\varepsilon}{2\varepsilon}+\left\lfloor\frac{x+\varepsilon}{4\varepsilon}\right\rfloor}\sin{\frac{\pi}{\varepsilon}(x-\varepsilon)}}{(x-\varepsilon)\pi}, &\text{if  } x\ne\varepsilon;\\
    \end{dcases} =\\
    &= \begin{dcases}
    -1, &\text{if  } x=\varepsilon,\\
    0, &\text{if  } x\ne\varepsilon.\\
    \end{dcases} = a_2(x,t,m,\varepsilon,u_{\varepsilon}).
\end{align*}
The equality before the last one holds because $x\in \varepsilon\mathbb{Z}$.

\textit{Step.} Consider the following 4 subcases.

\textit{Subcase 1:}  $t/\varepsilon \equiv_4 2$.
Then for $a_1(x,t,m,\varepsilon, u_{\varepsilon})$ we have the following chain of equalities:
\begin{align*}
    &a_1(x,t,m,\varepsilon, u_{\varepsilon}) = \frac{a_1(x+\varepsilon,t-\varepsilon,m,\varepsilon, u_{\varepsilon}) + m\varepsilon a_2(x+\varepsilon,t-\varepsilon,m,\varepsilon, u_{\varepsilon})}{\sqrt{1+m^2\varepsilon^2}} =\\
    &=\frac{(-1)^{\frac{x}{2\varepsilon}+\left\lfloor\frac{x+t}{4\varepsilon}\right\rfloor}}{\sqrt{1+m^2\varepsilon^2}}\left(\frac{m\varepsilon^2}{2\pi(1+m^2\varepsilon^2)}\int\limits_{-\pi/\varepsilon}^{\pi/\varepsilon}2 e^{ip(x+\varepsilon)}\sin\omega_p(t-2\varepsilon)\cdot\frac{\sin p\varepsilon}{\cos2\omega_p\varepsilon}dp \hspace{0.2cm}+ \right.\\
    &\left.+\frac{m\varepsilon^2}{2\pi}\int\limits_{-\pi/\varepsilon}^{\pi/\varepsilon}e^{ipx}\left(i\sin\omega_p(t-2\varepsilon)\cdot\frac{m^2\varepsilon^2+\cos 2p\varepsilon}{(1+m^2\varepsilon^2)\cos2\omega_p\varepsilon}-\cos\omega_p(t-2\varepsilon)\right)dp \right) =\\
    &=\frac{(-1)^{\frac{x}{2\varepsilon}+\left\lfloor\frac{x+t}{4\varepsilon}\right\rfloor}m\varepsilon^2}{2\pi\sqrt{1+m^2\varepsilon^2}}\cdot\\
    &\cdot\int\limits_{-\pi/\varepsilon}^{\pi/\varepsilon}e^{ipx}\left(\sin\omega_p(t-2\varepsilon)\cdot\frac{2e^{ip\varepsilon}\sin p\varepsilon + i(m^2\varepsilon^2+\cos 2p\varepsilon)}{(1+m^2\varepsilon^2)\cos2\omega_p\varepsilon} - \cos\omega_p(t-2\varepsilon)\right)dp=\\
    &=\frac{(-1)^{\frac{x}{2\varepsilon}+\left\lfloor\frac{x+t}{4\varepsilon}\right\rfloor}m\varepsilon^2}{2\pi\sqrt{1+m^2\varepsilon^2}}\cdot\\
    &\cdot\int\limits_{-\pi/\varepsilon}^{\pi/\varepsilon}e^{ipx}\left(i\sin\omega_p(t-2\varepsilon)\cdot\frac{1 - i\sin2\omega_p\varepsilon }{\cos2\omega_p\varepsilon} - \cos\omega_p(t-2\varepsilon)\right)dp =\\
    & = (-1)^{\frac{x}{2\varepsilon} + \left\lfloor\frac{x+t}{4\varepsilon}\right\rfloor}\frac{m\varepsilon^2}{2\pi\sqrt{1+m^2\varepsilon^2}}\int\limits_{-\pi/\varepsilon}^{\pi/\varepsilon}e^{ipx}\cdot\frac{i\sin{\omega_p(t-2\varepsilon) - \cos{\omega_pt} }}{\cos{2\omega_p\varepsilon}} dp.
\end{align*}
\normalsize
Here the first equality follows from Dirac's equation (Lemma~\ref{dirak}), because $u_{\varepsilon}(x+\varepsilon/2, t-\varepsilon/2)=1$ for integer $\frac{x+t}{2\varepsilon}$. The second one holds by the inductive hypothesis. The third one is obtained by expansion in $\sin\omega_p(t-2\varepsilon)$ and $\cos\omega_p(t-2\varepsilon)$ respectively. The fourth one follows from the following chain of equalities:
 \begin{align*}
     &\frac{2e^{ip\varepsilon} \sin p\varepsilon + i(m^2\varepsilon^2 + \cos 2p\varepsilon)}{1+m^2\varepsilon^2} = \frac{i(1 - e^{2ip\varepsilon} + m^2\varepsilon^2 + \cos 2p\varepsilon)}{1+m^2\varepsilon^2} = \\
     & =i\frac{1+m^2\varepsilon^2 - i\sin 2p\varepsilon}{1+m^2\varepsilon^2} = i(1 - i\sin 2\omega_p\varepsilon).
 \end{align*}
The fifth one follows from the formula for the cosine of the sum.
 
For $a_2(x,t,m,\varepsilon, u_{\varepsilon})$ we have the following chain of equalities:
\begin{align*}
    &a_2(x,t,m,\varepsilon, u_{\varepsilon}) = \\
    &=(-1)^{\frac{t-x-2\varepsilon}{2\varepsilon}}\frac{a_2(x-\varepsilon,t-\varepsilon,m,\varepsilon, u_{\varepsilon}) - m\varepsilon a_1(x-\varepsilon,t-\varepsilon,m,\varepsilon, u_{\varepsilon})}{\sqrt{1+m^2\varepsilon^2}} =\\
    &=\frac{(-1)^{\frac{x-2\varepsilon}{2\varepsilon}+\left\lfloor\frac{x+t - 2\varepsilon}{4\varepsilon}\right\rfloor + \frac{t-x-2\varepsilon}{2\varepsilon} }}{\sqrt{1+m^2\varepsilon^2}}\cdot\\
    &\cdot\left(\frac{\varepsilon}{2\pi}\int\limits_{-\pi/\varepsilon}^{\pi/\varepsilon}e^{ip(x-2\varepsilon)}\left(\frac{i\sin\omega_p(t-2\varepsilon)(m^2\varepsilon^2+\cos 2p\varepsilon)}{(1+m^2\varepsilon^2)\cos2\omega_p\varepsilon}-\cos\omega_p(t-2\varepsilon)\right)dp \right. -\\
    &-\left.m\varepsilon\frac{m\varepsilon^2}{2\pi(1+m^2\varepsilon^2)}\int\limits_{-\pi/\varepsilon}^{\pi/\varepsilon}2 e^{ip(x-\varepsilon)}\sin\omega_p(t-2\varepsilon)\cdot\frac{\sin p\varepsilon}{\cos2\omega_p\varepsilon}dp  \right)=\\
    &=\frac{(-1)^{\frac{x-2\varepsilon}{2\varepsilon}+\left\lfloor\frac{x+t}{4\varepsilon}\right\rfloor}\varepsilon}{2\pi\sqrt{1+m^2\varepsilon^2}}\cdot\\
    &\cdot\int\limits_{-\pi/\varepsilon}^{\pi/\varepsilon}e^{ip(x-2\varepsilon)}\left(\sin\omega_p(t-2\varepsilon)\cdot\frac{i(m^2\varepsilon^2+\cos 2p\varepsilon) - 2m^2\varepsilon^2e^{ip\varepsilon}\sin p\varepsilon}{(1+m^2\varepsilon^2)\cos2\omega_p\varepsilon}- \right. \\
    & \left.  -\cos\omega_p(t-2\varepsilon) \right)dp =\\
    &=\frac{(-1)^{\frac{x-2\varepsilon}{2\varepsilon}+\left\lfloor\frac{x+t }{4\varepsilon}\right\rfloor}\varepsilon}{2\pi\sqrt{1+m^2\varepsilon^2}}\cdot\\
    &\cdot\int\limits_{-\pi/\varepsilon}^{\pi/\varepsilon}e^{ip(x-2\varepsilon)}\left(i\sin\omega_p(t-2\varepsilon)\cdot\frac{\cos 2p\varepsilon + m^2\varepsilon^2e^{2ip\varepsilon}}{(1+m^2\varepsilon^2)\cos2\omega_p\varepsilon} - \cos\omega_p(t-2\varepsilon)\right)dp = \\
    &=(-1)^{\frac{x-2\varepsilon}{2\varepsilon}+\left\lfloor\frac{x+t}{4\varepsilon}\right\rfloor}\frac{\varepsilon}{2\pi\sqrt{1+m^2\varepsilon^2}}\int\limits_{-\pi/\varepsilon}^{\pi/\varepsilon}e^{ip(x-2\varepsilon)}\cdot\frac{ie^{2ip\varepsilon}\sin{\omega_p(t-2\varepsilon)} - \cos{\omega_pt} }{\cos{2\omega_p\varepsilon}} dp.
\end{align*}
\normalsize
Here the first equality follows from Dirac's equation (Lemma~\ref{dirak}), because $u_{\varepsilon}(x-\varepsilon/2, t-\varepsilon/2)=(-1)^{\frac{t-x-2\varepsilon}{2\varepsilon}}$. The second one holds by the inductive hypothesis. The third one is obtained by expansion in $\sin\omega_p(t-2\varepsilon)$ and $\cos\omega_p(t-2\varepsilon)$ respectively. The sign in the third equality is obtained from the assumption that $t/\varepsilon\equiv_{4}2$ and the following chain of equalities:
\begin{align*}
 \left\lfloor\frac{x+t - 2\varepsilon}{4\varepsilon}\right\rfloor + \frac{t-x-2\varepsilon}{2\varepsilon}   \equiv_2 
 \begin{dcases}
 \frac{x+t-2\varepsilon}{4\varepsilon}, &\text{if }\frac{x}{\varepsilon}\equiv_{4} 0,\\
 \frac{x+t}{4\varepsilon}, &\text{if }\frac{x}{\varepsilon}\equiv_{4} 2;
 \end{dcases} = \left\lfloor\frac{x+t}{4\varepsilon}\right\rfloor.
\end{align*} 
The fourth one is obtained by expressing $\sin p\varepsilon$ through $e^{\pm ip\varepsilon}$. The fifth one follows from the following chain of equalities:
\begin{align*}
    &i\sin\omega_p(t-2\varepsilon)\cdot\frac{\cos 2p\varepsilon + m^2\varepsilon^2e^{2ip\varepsilon}}{(1+m^2\varepsilon^2)\cos2\omega_p\varepsilon} - \cos\omega_p(t-2\varepsilon) =\\
    &= i\sin\omega_p(t-2\varepsilon)\cdot\frac{e^{2ip\varepsilon} -i\sin 2\omega_p\varepsilon}{\cos2\omega_p\varepsilon} - \cos\omega_p(t-2\varepsilon) =\\
    &=\frac{ie^{2ip\varepsilon}\sin\omega_p(t-2\varepsilon) -\cos\omega_pt}{\cos2\omega_p\varepsilon}.
\end{align*}
Thus, the first subcase is proved.

\textit{Subcase 2:} $t/\varepsilon \equiv_4 3$. Then for $a_1(x,t,m,\varepsilon, u_{\varepsilon})$ we have the following chain of equalities:
\begin{align*}
    &a_1(x,t,m,\varepsilon, u_{\varepsilon}) = \frac{a_1(x+\varepsilon,t-\varepsilon,m,\varepsilon, u_{\varepsilon}) + m\varepsilon a_2(x+\varepsilon,t-\varepsilon,m,\varepsilon, u_{\varepsilon})}{\sqrt{1+m^2\varepsilon^2}} = \\
    &=\frac{1}{\sqrt{1+m^2\varepsilon^2}}\cdot\\
    &\cdot\left((-1)^{\frac{x+\varepsilon}{2\varepsilon} + \left\lfloor\frac{x+t}{4\varepsilon}\right\rfloor}\frac{m\varepsilon^2}{2\pi\sqrt{1+m^2\varepsilon^2}}\int\limits_{-\pi/\varepsilon}^{\pi/\varepsilon}e^{ip(x+\varepsilon)}\cdot\frac{i\sin{\omega_p(t-3\varepsilon) - \cos{\omega_p(t-\varepsilon)} }}{\cos{2\omega_p\varepsilon}} dp + \right.\\
   &\left.+\frac{(-1)^{\frac{x-\varepsilon}{2\varepsilon}+\left\lfloor\frac{x+t }{4\varepsilon}\right\rfloor}m\varepsilon^2}{2\pi\sqrt{1+m^2\varepsilon^2}}\int\limits_{-\pi/\varepsilon}^{\pi/\varepsilon}e^{ip(x-\varepsilon)}\cdot\frac{ie^{2ip\varepsilon}\sin{\omega_p(t-3\varepsilon)} - \cos{\omega_p(t-\varepsilon)} }{\cos{2\omega_p\varepsilon}} dp\right) =\\
   &=\frac{(-1)^{\frac{x+\varepsilon}{2\varepsilon}+\left\lfloor\frac{x+t }{4\varepsilon}\right\rfloor}m\varepsilon^2}{2\pi(1+m^2\varepsilon^2)}\int\limits_{-\pi/\varepsilon}^{\pi/\varepsilon}e^{ip(x-\varepsilon)}\cdot\frac{(1 - e^{2ip\varepsilon})\cos \omega_p(t-\varepsilon)}{\cos 2\omega_p\varepsilon}=\\
   &=(-1)^{\frac{x-\varepsilon}{2\varepsilon}+\left\lfloor\frac{x+t }{4\varepsilon}\right\rfloor}\frac{im\varepsilon^2}{\pi(1+m^2\varepsilon^2)}\int\limits_{-\pi/\varepsilon}^{\pi/\varepsilon}e^{ipx}\cos{\omega_p(t-\varepsilon)}\cdot\frac{\sin{p\varepsilon}}{\cos{2\omega_p\varepsilon}}dp.
\end{align*}
\normalsize
The first equality holds by Dirac's equation (Lemma~\ref{dirak}), because $u_{\varepsilon}(x+\varepsilon/2, t-\varepsilon/2)=1$ for integer $\frac{t+x}{2\varepsilon}$. The second one is true by the inductive hypothesis. The third one is obtained by expansion in $\sin\omega_p(t-3\varepsilon)$ and $\cos\omega_p(t-\varepsilon)$ respectively. The fourth one follows from the expression of $\sin p\varepsilon$ through $e^{\pm ip\varepsilon}$.

For $a_2(x,t,m,\varepsilon, u_{\varepsilon})$ we have the following chain of equalities:
\begin{align*}
    &a_2(x,t,m,\varepsilon, u_{\varepsilon}) =\\
    &=(-1)^{\frac{t-x-2\varepsilon}{2\varepsilon}}\frac{a_2(x-\varepsilon,t-\varepsilon,m,\varepsilon, u_{\varepsilon}) - m\varepsilon a_1(x-\varepsilon,t-\varepsilon,m,\varepsilon, u_{\varepsilon}) }{\sqrt{1+m^2\varepsilon^2}} =\\ &=\frac{(-1)^{\frac{t-x-2\varepsilon}{2\varepsilon}+\left\lfloor\frac{x+t - 2\varepsilon}{4\varepsilon}\right\rfloor}}{\sqrt{1+m^2\varepsilon^2}}\cdot\\
    &\cdot\left(\frac{(-1)^{\frac{x-3\varepsilon}{2\varepsilon}}\varepsilon}{2\pi\sqrt{1+m^2\varepsilon^2}}\int\limits_{-\pi/\varepsilon}^{\pi/\varepsilon}e^{ip(x-3\varepsilon)}\cdot\frac{ie^{2ip\varepsilon}\sin{\omega_p(t-3\varepsilon)} - \cos{\omega_p(t-\varepsilon)} }{\cos{2\omega_p\varepsilon}} dp \hspace{0.2cm}- \right.\\
   &\left.- \frac{(-1)^{\frac{x-\varepsilon}{2\varepsilon}}m^2\varepsilon^3}{2\pi\sqrt{1+m^2\varepsilon^2}}\int\limits_{-\pi/\varepsilon}^{\pi/\varepsilon}e^{ip(x-\varepsilon)}\cdot\frac{i\sin{\omega_p(t-3\varepsilon) - \cos{\omega_p(t-\varepsilon)} }}{\cos{2\omega_p\varepsilon}} dp\right) =\\
   &=\frac{(-1)^{\frac{x-\varepsilon}{2\varepsilon}+\left\lfloor\frac{x+t }{4\varepsilon}\right\rfloor}\varepsilon}{2\pi(1+m^2\varepsilon^2)}\cdot\\
   &\cdot\int\limits_{-\pi/\varepsilon}^{\pi/\varepsilon}e^{ip(x-\varepsilon)}\frac{i\sin{\omega_p(t-3\varepsilon)}(1+m^2\varepsilon^2) - (m^2\varepsilon^2 + e^{-2ip\varepsilon})\cos{\omega_p(t-\varepsilon)}}{\cos{2\omega_p\varepsilon}} dp =\\
   &=\frac{(-1)^{\frac{x+\varepsilon}{2\varepsilon}+\left\lfloor\frac{x+t }{4\varepsilon}\right\rfloor}\varepsilon }{2\pi}\cdot\\
   &\cdot\int\limits_{-\pi/\varepsilon}^{\pi/\varepsilon}e^{ip(x-\varepsilon)}\left(\cos{\omega_p(t-\varepsilon)}\cdot\frac{m^2\varepsilon^2 +\cos{2p\varepsilon}}{(1+m^2\varepsilon^2)\cos{2\omega_p\varepsilon}} - i\sin{\omega_p(t-\varepsilon)}\right) dp.
\end{align*}
\normalsize
The first equality follows from the Dirac equation (Lemma~\ref{dirak}), because $u_{\varepsilon}(x-\varepsilon/2, t-\varepsilon/2)=(-1)^{\frac{t-x-2\varepsilon}{2\varepsilon}}$ for integer $\frac{t-x}{2\varepsilon}$. The second one is true by the inductive hypothesis. The third one is obtained by expansion in $\sin\omega_p(t-3\varepsilon)$ and $\cos\omega_p(t-\varepsilon)$ respectively. The sign in the third equality follows from the assumption that $t/\varepsilon\equiv_{4}3$ and the following chain of equalities:
\begin{align*}
 \left\lfloor\frac{x+t - 2\varepsilon}{4\varepsilon}\right\rfloor + \frac{t-x}{2\varepsilon}   \equiv_2 
 \begin{dcases}
 \frac{x+t}{4\varepsilon}, &\text{if }\frac{x}{\varepsilon}\equiv_{4} 1,\\
 \frac{x+t-2\varepsilon}{4\varepsilon}, &\text{if }\frac{x}{\varepsilon}\equiv_{4} 3;
 \end{dcases} = \left\lfloor\frac{x+t}{4\varepsilon}\right\rfloor.
\end{align*}
The fourth one follows from the following chain of equalities:
\begin{align*}
   & \frac{i\sin{\omega_p(t-3\varepsilon)}(1+m^2\varepsilon^2) - (m^2\varepsilon^2 + e^{-2ip\varepsilon})\cos{\omega_p(t-\varepsilon)}}{(1+m^2\varepsilon^2)\cos{2\omega_p\varepsilon}} = \\
    &=i\sin{\omega_p(t-\varepsilon)} - \frac{(i\sin2p\varepsilon + m^2\varepsilon^2 + e^{-2ip\varepsilon})\cos{\omega_p(t-\varepsilon)}}{(1+m^2\varepsilon^2)\cos{2\omega_p\varepsilon}} = \\
    &=i\sin{\omega_p(t-\varepsilon)} - \frac{( m^2\varepsilon^2 + \cos{2p\varepsilon})\cos{\omega_p(t-\varepsilon)}}{(1+m^2\varepsilon^2)\cos{2\omega_p\varepsilon}},
\end{align*}
where we used the formula for the sine of the difference, definition of $\omega_p$, and the Euler formula for $e^{-2ip\varepsilon}$.
This completes the proof of subcase 2. 

\emph{Subcases 3 and 4}: $\frac{t}{\varepsilon}\equiv_4 0 $ and $\frac{t}{\varepsilon} \equiv_4 1$ are proved analogously to subcases 1 and 2, only $\cos \omega_p(t-k\varepsilon)$ is replaced by $\sin\omega_p(t-k\varepsilon)$ and vice versa, where $k=2$ or $3$. \\
$\square$

\end{document}